\documentclass[journal,10pt]{IEEEtran}

\usepackage{amsmath}

\usepackage{amsthm}
\usepackage{amsfonts}
\usepackage{dsfont} %For some reason amsfonts doesn't let me compile \mathbbm{1}. %%bbm had type 3 fonts
\usepackage{graphicx}
\usepackage{subfigure}
\usepackage{cite}
\usepackage{nicefrac}
\usepackage{bbm}
\usepackage{algorithm,algorithmic}
\usepackage{amssymb}
\usepackage{siunitx}
\usepackage{array}
\usepackage{upgreek}

\newtheorem{defn}{Definition}
\newtheorem{thm}{Theorem}
\newtheorem{lem}{Lemma}
\newtheorem{prop}{Proposition}
\newtheorem{cor}{Corollary}

\newtheorem{rem}{Remark}

\newtheorem{assump}{Assumption}
\newtheorem{prob}{Problem}

\newcommand{\e}{\varepsilon}
\newcommand{\G}{\mathcal{G}}

\newcommand{\V}{\mathbb{V}}
\newcommand{\E}{\mathcal{E}}
\newcommand{\EE}{\mathbb{E}}

\newcommand{\R}{\mathbb{R}}
\newcommand{\PP}{\mathbb{P}}
\newcommand{\Id}{\mathrm{Id}}

\usepackage[color=yellow, textsize=small, textwidth=\marginparwidth, draft]{todonotes}   % notes showed
\IEEEoverridecommandlockouts
\definecolor{forestgreen}{rgb}{0.13, 0.55, 0.13}

\begin{document}
\title{Network Identification for Diffusively-Coupled Networks with Minimal Time Complexity}
\author{Miel Sharf~%~\IEEEmembership{Member,~IEEE} 
and Daniel Zelazo,~\IEEEmembership{Senior~Member,~IEEE,}
\thanks{M. Sharf is with the Department of Decision and Control Systems, KTH Royal Institute of Technology, Stockholm, Sweden as well as Digital Futures {\tt \small sharf@kth.se}. D. Zelazo is with the Faculty of Aerospace Engineering, Israel Institute of Technology, Haifa, Israel.
%\thanks{M. Sharf is with Jether Energy Research, Tel Aviv, Israel {\tt \small mielsharf@gmail.com}. D. Zelazo is with the Faculty of Aerospace Engineering, Israel Institute of Technology, Haifa, Israel.
    {\tt\small dzelazo@technion.ac.il}.  }
}%\author

\maketitle

\begin{abstract}
The theory of network identification, namely identifying the (weighted) interaction topology among a known number of agents, has been widely developed for linear agents. However, the theory for nonlinear agents using probing inputs is far less developed, relying on dynamics linearization, \textcolor{black}{and thus cannot be applied to networks with non-smooth or discontinuous dynamics}. We use global convergence properties of the network, which can be assured using passivity theory, to present a network identification method for nonlinear agents. We do so by linearizing the steady-state equations rather than the dynamics, achieving a sub-cubic time algorithm for network identification. We also study the problem of network identification from a complexity theory standpoint, showing that the presented algorithms are optimal in terms of time complexity. We  demonstrate the presented algorithm in two case studies \textcolor{black}{ with discontinuous dynamics}.
\end{abstract}

\vspace{-10pt}
\section{Introduction}\label{sec.Intro}
Multi-agent networks (MANs) have been in the pinnacle of research in the last few years, both for their variety of applications and their deep theoretical framework. They have been applied in a wide range of domains, including formation control and distributed computing \cite{Mesbahi2010}. One of the most important aspects in MANs is the information-exchange layer, governing the interaction of agents with one another. Identifying the underlying network of a MAN from measurements is of great importance in many domains, including privacy \cite{Zheleva2012}, \textcolor{black}{(mis-)information spread \cite{Zhang2021}, epidemiological models \cite{Prasse2020}}, and brain connectivity in neuroscience \cite{Sakkalis2011}.

Network identification takes many different forms in the literature. The two main paradigms for network identification are topology reconsturction \cite{Materassi2012} and single-operator identification \textcolor{black}{\cite{Dankers2015,Bazanella2019}}. Different problem definitions also vary by the level of interaction with the network and the available data. First, network identification can be based on unmeasured intrinsic persistent excitations \cite{Chu2008}. Second, network identification can be based on naturally occurring, but measured excitations \cite{Yuan2011}. Third, network identification can be based on interacting with the network, e.g. by injecting specially designed inputs \textcolor{black}{\cite{Timme2007,vanWaarde2019, Sharf2018b}} or by node knockout \cite{NabiAbdolyousefi2012,NabiAbdolyousefi2014}. 
Seeveral methods have been applied to network identification, including spectral methods \cite{Mauroy2017}, information-theoretic methods such as Granger causality \cite{Vicente2011,Bressler2011,Materassi2010}, compressed sensing and sparsity-promoting optimization \textcolor{black}{\cite{Fonken2022,Gardner2003,Guangjun2015,Zhang2021}}, and Wiener filtering \cite{Materassi2012}.

The network identification literature for nonlinear agents and/or interactions is not as developed. Some works linearize the network dynamics and then use sparse recovery to identify the connecting matrix \cite{Gardner2003,Guangjun2015,Prasse2020}. Other works use adaptive observers to find the network structure, while assuming the Lipschitzity of certain elements in the dynamics \cite{Chen2009,Burbano2019}. \textcolor{black}{However, these methods are not applicable when the dynamics are not Lipschitz continuous, let alone discontinuous. For example, dry friction gives a nonlinear discontinuous term in the dynamics \cite{vanderLinden1993}, and some finite-time consensus protocols also include nonlinear non-Lipschitz dynamics \cite{Shi2018}.
Moreover, existing methods measure the network state in constant intervals, resulting in an unnecessary power expenditure for networks with extremely fast dynamics, e.g. networks following a finite-time consensus protocol. It will also be wasteful for relatively-static networks, e.g., networks of autonomous vehicles, especially in platooning. There, one desires to keep the vehicles' position and velocity very close to a certain desired steady-state, either for safety or for efficiency.}

In this paper, we present a method for network identification in the graph identification framework while interacting with the agents in the network by injecting exogenous inputs. Our approach relies on the asymptotic stability of the network with respect to constant exogenous inputs, which can be verified using passivity theory, a well-known tool used to study MANs \cite{Burger2014,Sharf2018a,Sharf2019a}. Namely, \cite{Sharf2018a} shows a connection between the exogenous input of a diffusively-coupled MAN and its steady-state output.
We aim to use this connection, together with the global convergence properties of the network, to provide a network identification scheme for MANs. We do so by injecting appropriately-defined constant exogenous inputs and tracking the output of the agents. The key idea is that the steady-state outputs are one-to-one dependent on the exogenous input. This dependency can be linearized, and the associated matrix can be found by considering a finite number of inputs and outputs. Our contributions are stated as follows:
\begin{itemize}
\item[i)] We present a sub-cubic network identification algorithm for globally convergent networks. The identification is exact for noise-less LTI networks, and an error bound for the identified network is derived for the general case.
\item[ii)] We discuss the robustness of the algorithm, its assumptions, and compare it to other methods in the literature. 
\item[iii)] We explore the complexity theory behind network identification algorithms, and prove that the algorithm we presented is optimal in terms of time complexity.
\end{itemize}

The paper is organized as follows. Section \ref{sec.background} surveys relevant details about diffusively-coupled networks and their convergence. Section \ref{sec.ProblemFormulation} presents the problem formulation. Section \ref{sec.GeneralAlgorithm} presents the network identification algorithm. Section \ref{sec.Discussion} discusses the assumptions of the algorithm. Section \ref{sec.Complexity} studies network identification from a complexity theory standpoint, showing that the presented algorithm is optimal. Section \ref{sec.CaseStudy} demonstrates the algorithm in two case studies.

\paragraph*{Notation}
Time-dependent signals are denoted with italic letters, e.g. $y=y(t)$, and constant signals are denoted in an upright font, e.g. $\mathrm y$.
We use basic notions from algebraic graph theory. An undirected graph $\mathcal{G}=(\mathbb{V},\mathbb{E})$ consists of a finite set of vertices $\mathbb{V}$ and edges $\mathbb{E} \subset \mathbb{V} \times \mathbb{V}$. The edge with ends $i,j\in \V$ is denoted as $e=\{i,j\}$. For each edge $e$, we pick an arbitrary orientation and denote $e=(i,j)$.
The incidence matrix of $\mathcal{G}$, denoted \textcolor{black}{$\mathcal{E}_\mathcal{G}\in\mathbb{R}^{|\mathbb{V}|\times|\mathbb{E}|}$}, is defined such that for any edge $e=(i,j)\in \mathbb{E}$, $[\mathcal{E_\G}]_{ie} =+1$, $[\mathcal{E_\G}]_{je} =-1$, and $[\mathcal{E_\G}]_{\ell e} =0$ for $\ell \neq i,j$. The Laplacian of graph $\G$ is the matrix $\E_\G \E_\G^\top$, and a weighted Laplacian is the matrix $\E_\G D \E_\G^\top$ for a diagonal matrix $D>0$.
Furthermore, a weighted graph is a pair $\G_\nu = (\G,\{\nu_e\}_e)$ where $\G$ is a graph and $\{\nu_e\}$ are real numbers assigned to the edges. 

We also use basic notions from linear algebra. The standard basis vectors in $\mathbb{R}^n$ are denoted $\mathrm e_1,\ldots,\mathrm e_n$. The kernel of a linear map $T:X\to Y$ is denoted by $\ker{T}$. Furthermore, if $U$ is a subspace of an inner-product space $X$ (e.g., $\mathbb{R}^d$), we denote the orthogonal complement of $U$ by $U^{\perp}$. We write $A > 0$ ($A\ge0$) for a positive (semi-)definite matrix $A$. Moreover, if $A>0$, we denote the minimal eigenvalue of $A$ by $\underline{\sigma}(A)$. \textcolor{black}{We let $\mathbbm{1}_n \in \R^n$ be the all-ones vector}.
A matrix $M$ is called elementary if the matrix multiplication map $A \mapsto MA$ defines one of the following operations: switching two rows in $A$, scaling a row of $A$, or adding a row of $A$ to another.
We shall also employ basic notations from complexity theory. For two functions $f,g:\mathbb{N}\to\mathbb{R}$, if there exist constants $c,n_0>0$ such that $f(n) \le cg(n)$ holds for all integers $n\ge n_0$, we write $f = O(g)$. If the reverse inequality holds, we write $f = \Omega(g)$ instead.
Lastly, we say that a signal $a(t)$ is in $\mathcal{C}^q$ if it is continuously differentiable $q$ times.

%\section{Background: Network Optimization, MEIP Multi-Agent Systems and Complexity of Matrix Multiplication}\label{sec.background}
\vspace{-4pt}
\section{Preliminaries}
\label{sec.background}

In this section, we provide some needed background on diffusively-coupled networks and complexity theory.

\vspace{-6pt}
\subsection{Diffusively-Coupled Networks and Steady-States}
Consider a collection of agents interacting over a network $\mathcal{G}=(\mathbb{V},\mathbb{E})$.  Assign to each node $i\in \mathbb{V}$ (the agents) and each edge $e \in \mathbb{E}$ (the controllers) the following dynamical systems
\begin{align} \label{eq.MAN}
	\Sigma_i:
	\begin{cases}
		\dot{x}_i = f_i(x_i) + q_i(x_i)u_i \\
		y_i = h_i(x_i)
	\end{cases}, \,
	\textcolor{black}{\Pi_e: \mu_e = g_e(\zeta_e)}
	%\left\{\begin{array}{c}
	%	\dot{\eta}_e = \phi_e(\eta_e,\zeta_e) \\
	%	\mu_e = \psi_e(\eta_e,\zeta_e)
	%	\end{array} \right. .
\end{align}
We stack $y=[y_1^\top,\ldots,y_{|\mathbb{V}|}^\top]^\top$ and similarly for $u,\zeta$ and $\mu$ and the operators $\Sigma$ and $\Pi$ and the function $g$. The network is diffusively coupled by defining the control input as $u = -\E_\G\mu$, and the controller input as $\zeta = \E_\G^\top y$. The closed-loop network is denoted by $(\G,\Sigma,\Pi)$, and is illustrated in Fig. \ref{ClosedLoop}.

\begin{figure} [!t] 
    \centering
    \includegraphics[scale=0.28]{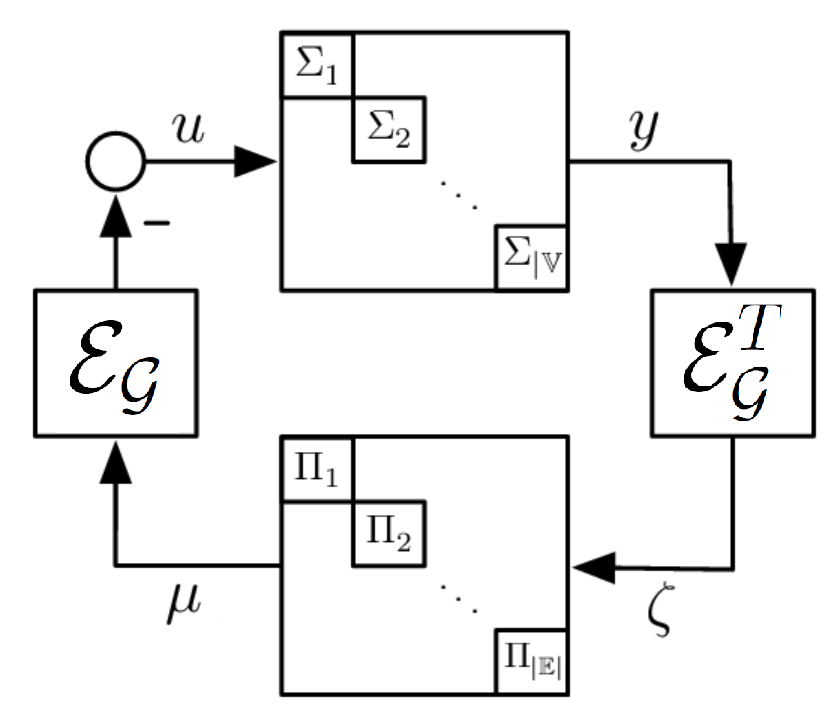}
    \vspace{-5pt}
    \caption{Block-diagram of the closed loop.}
    \label{ClosedLoop}
    \vspace{-10pt}
\end{figure}

We consider steady-states of such networks. If $(\mathrm{u},\mathrm{y},\upzeta,\upmu)$ is a steady-state of the network, then $(\mathrm{u}_i,\mathrm{y}_i)$ is a steady-state input-output pair of the $i$-th agent, and $(\upzeta_e,\upmu_e)$ is a steady-state input-output pair of the $e$-th edge. This motivates the following definition, originally introduced in \cite{Burger2014}.
\begin{defn}
The \emph{steady-state input-output relation} $k$ of a dynamical system is the collection of all constant steady-state input-output pairs of the system. Given a steady-state input $\mathrm{u}$ and a steady-state output $\mathrm{y}$, we define $k(\mathrm u) = \{\mathrm y:\ (\mathrm{u,y})\in k\}$ and $k^{-1}(\mathrm y) = \{\mathrm u:\ (\mathrm{u,y})\in k\}$.
\end{defn}
Let $k_i$ be the steady-state relations for the $i$-th agent and let $k$ be their stacked version. It is shown in \cite{Sharf2017} that $(\mathrm{u},\mathrm{y},\upzeta,\upmu)$ is a steady-state of the network if and only if $\mathrm y\in k(\mathrm u)$, $\upzeta=\E_\G^\top\mathrm y$, $\upmu\in g(\upzeta)$, and $\mathrm u=-\E_\G\upmu$ all hold. Consequently, $\mathrm{y}$ is a steady-state for the network $(\mathcal{G},\Sigma, \Pi)$ if and only if
\begin{align}\label{eq.SteadyStateEquation}
0\in k^{-1}(\mathrm y) + \E_\G g(\E_\G^\top\mathrm y).
\end{align}
% 
%This expression summarizes the constraints that must be satisfied by the network to achieve a steady-state solution.

\vspace{-10pt}
\subsection{Tools from Passivity and Complexity Theory}\label{subsec.passivity_complexity}
Our main tool will be a variant of \eqref{eq.SteadyStateEquation}, relating constant exogenous inputs and the corresponding steady-state outputs. This connection will only be useful if we can measure the steady-state of the network, i.e., we must assume the network converges, which is guaranteed under an \textcolor{black}{(output-strict)} passivity\footnote{\textcolor{black}{Recall that the system $\Sigma$ is output-strictly passive with respect to the steady-state input-output pair $(\mathrm{u},\mathrm y)$ if there exists a positive-definite storage function $S(x)$ and a constant $\rho > 0$ such that any trajectory satisfies $\frac{S(x(t))}{dt} \le (u - \mathrm u)(y - \mathrm y) - \rho(y - \mathrm y)^2$.}} assumption on the agents and controllers \cite{Burger2014}.

\begin{defn}[\small{Maximal Equilibrium Independent Passivity \cite{Burger2014}}]
Consider the dynamical system $\Upsilon$ defined by $\dot{x} = f(x,u),~
y = h(x,u)$,
with steady-state input-output relation $r$.
The system $\Upsilon$ is said to be \emph{(output-strictly) maximal equilibrium independent passive}, or MEIP, if the following conditions hold:
\begin{enumerate}
\item[i)] The system $\Upsilon$ is (output-strictly) passive with respect to any steady state pair $(\mathrm u,\mathrm y) \in r$.
\item[ii)] The relation $r$ is maximally monotone. That is, if $(\mathrm u,\mathrm y),(\mathrm u^\prime,\mathrm y^\prime)\in r$ then $(\mathrm u - \mathrm u^\prime)(\mathrm y - \mathrm y^\prime) \ge 0$, and $r$ is not contained in a larger monotone relation.
\end{enumerate}
\end{defn}
Such systems include single integrators, gradient systems, port-Hamiltonian systems, and others. The monotonicity requirement is used to prove existence of a closed-loop steady-state. See \cite{Burger2014} or \cite{Sharf2018a} for more details and examples.  

\begin{thm}[\hspace{-0.1pt}\cite{Burger2014,Sharf2018a}]\label{thm.ClosedLoopSteadyStates}
Consider the network $(\G,\Sigma,\Pi)$. Assume all agents are MEIP, and all controllers are output-strictly MEIP, or vice versa. Then the signals $u,y,\zeta,\mu$ converge to some steady-state values ${\mathrm u},{\mathrm y},{\mathrm \upzeta},{\mathrm \upmu}$ satisfying \eqref{eq.SteadyStateEquation}.
\end{thm}

If we measure the steady-state outputs and the corresponding constant exogenous inputs, we end up with a function inversion problem stemming from a variant of \eqref{eq.SteadyStateEquation}. Linearizing it results in a matrix inversion problem instead, whose computational complexity will now be discusses. Matrix inversion and matrix multiplication are known to have the same time complexity, which is usually denoted as $O(n^\omega)$, neglecting poly-logarithmic terms \cite{Cormen2009}. While schoolbook multiplication gives $\omega = 3$, other implementable algorithms yield $\omega \approx 2.807$ \cite{Strassen1969}, while some theoretical algorithms can give $\omega \approx 2.3728639$ \cite{LeGall2014}. The latter are impractical as the constant in front of $n^\omega$ is extremely large. We focus on inverting positive-definite matrices, with a time complexity denoted by $O(n^{\omega_1})$, similarly neglecting poly-logarithmic terms, where $2\le \omega_1 \le \omega < 3$, as reading the input takes $O(n^2)$ time.

\vspace{-9pt}
\section{Motivation and Problem Formulation}
\vspace{-5pt}
\label{sec.ProblemFormulation}
Our goal is to solve a network identification problem for MANs having the following (possibly nonlinear) dynamics:

\vspace{-7pt}
\small
\begin{align} \label{eq.System}
\dot{x}_i = f_i(x_i) + q_i(x_i)\left[\sum_{\{i,j\}\in\mathbb{E}} \nu_{ij}g_{ij}(h_j(x_j) - h_i(x_i)) + B_iw_i\right]
\end{align}\normalsize
where \textcolor{black}{the state is $x_i\in \R^{n_i}$, $f_i,q_i: \R^{n_i} \to \R^{n_i}$, $g_{ij}: \R\to \R$ and $h_i : \R^{n_i} \to \R$} are (not necessarily smooth) functions\footnote{The functions $g_{ij}$ are defined for all pairs, even those absent from the underlying graph. It is often assumed in MANs that each agent knows how to run a given protocol (i.e., consensus).}, and \textcolor{black}{ $B_i, w_i \in \R$}. These dynamics are achieved from the model \eqref{eq.MAN} by adding an exogenous input to the agents, now governed by the dynamics $\dot{x}_i = f_i(x_i) + q_i(x_i)(u_i + B_iw_i)$, as well as coupling strengths $\{\nu_{ij\}}$. Examples of MANs governed by \eqref{eq.System}, include the consensus protocol, the Kuramoto model for synchronizing oscillators \cite{Dorfler2014}, and traffic models \cite{Bando1995}.

In many models, only certain agents can be injected with an exogenous input (i.e., $B_i=0$ is possible), and only the output of certain agents can be measured. We assume all agents are both susceptible to exogenous inputs and measurable. The case in which only some agents are susceptible to exogenous inputs is tackled in Section \ref{sec.Discussion}. For now, we assume without loss of generality that $B_i = 1$, and denote \textcolor{black}{$N=\mathrm{diag}(\nu_{ij})_{\{i,j\} \in \EE} $}.
Denoting the network \eqref{eq.System} as $(\G_\nu,\Sigma,g)$, we make the following assumptions on the agents and controllers, allowing us to use the framework presented in Section \ref{sec.background}.  

\begin{assump} \label{Assumption1}
The closed-loop network $(\G_\nu,\Sigma,\Pi)$ converges to a constant steady-state for any constant exogenous input $\mathrm w$ and any initial condition.
\end{assump}
\begin{assump} \label{Assumption2}
The inverse of the steady-state input-output relation for each agent, $k_i^{-1}(\mathrm y_i)$, is a continuous monotone function of $\mathrm y_i$. Moreover, $g_{ij}(\zeta_{ij})$ is a continuous monotone function of $\zeta_{ij}$. Assume that on any bounded interval, there are at most finitely many points at which $k_i^{-1}$ and $g_{ij}$ are not twice differentiable. Moreover, assume that the set of points in $\R$ on which the derivative $\frac{dg_{ij}}{d\zeta_{ij}}$ nulls is of measure zero.
\end{assump}
Assumption \ref{Assumption1} is satisfied if the weights $\nu_{ij}$ are non-negative, if the systems $\{\Sigma_i\}_{i\in \V}$ are MEIP, the controllers $\{\Pi_e\}_{e\in \E}$ are output-strictly MEIP (or vice versa). However, this assumption is weaker, as convergence can be established using methods other than passivity. Note that the last part of Assumption \ref{Assumption2} holds if and only if the derivative of $g_{ij}$ is positive almost everywhere. We will note that if the agents are linear time invariant (LTI) then $k^{-1}_i$ is a linear function.
We now formulate the network identification problem considered in this paper.

\begin{prob} \label{prob.network_detection2}
Consider the network $(\G_\nu, \Sigma, g)$ of the form \eqref{eq.System} satisfying Assumptions \ref{Assumption1} and \ref{Assumption2}. Assume the steady-state input-output relations for the agents and controllers are known\footnote{\textcolor{black}{As seen later in Remark 3, we actually only need to know the derivative of $g_{ij}$ at a single point.}}, but the network structure $\G$ and coupling coefficients $\{\nu_{ij}\}$ are unknown.  Design the control inputs $w_i$ such that, together with measurements of the output of the network, it is possible to identify the graph $\G$ and the coupling coefficients $\{\nu_{ij}\}$. 
\end{prob}

The paper \cite{Sharf2018b} proposed a solution for Problem \ref{prob.network_detection2} under an MEIP assumption by feeding the network with a single \emph{constant} input and relying on a variant of \eqref{eq.SteadyStateEquation}. While the time complexity for LTI networks was $O(n^3)$, the runtime for general nonlinear networks could be superexponential, namely $O(2^{2^n})$. We wish to improve on this complexity by considering multiple constant exogenous inputs denoted as $\textcolor{black}{\mathrm w^{(i)}}$. As in \cite{Sharf2018b}, we write an equation connecting the steady-state output $\mathrm y$ to the constant exogenous input $\mathrm w$.

\begin{prop} \label{prop.ClosedLoopSteadyState}
Suppose that the network $(\G_\nu,\Sigma,g)$ is run with the constant exogenous input $\mathrm w$, and let $k$ be the steady-state relation for the agents. Then the output $y$ of the network converges to a steady-state $\mathrm y$ satisfying the following equation,
\begin{align}\label{eq.ClosedLoopEquation}
\mathrm{w} = k^{-1}(\mathrm y) + \E_\G N g(\E_\G^\top\mathrm{y}).
\end{align}
\end{prop}
\begin{proof}
The closed-loop network converges by Assumption \ref{Assumption1}. The equation for the steady-state output $\mathrm y$ of the closed-loop network follows from $\upmu = Ng(\mathrm \upzeta$), $k^{-1}(\mathrm y)=\mathrm u+\mathrm w$, $\upzeta=\E_\G^\top\mathrm y$, and $\mathrm u=-\E_\G\upmu$. This is an equality rather than an inclusion due to Assumption \ref{Assumption2}.
\end{proof}

Equation \eqref{eq.ClosedLoopEquation} shows that the steady-state output $\mathrm y$ depends not only on the constant exogenous input $\mathrm w$, but also on the incidence matrix $\E_\G$ and the weights $N = \mathrm{diag}(\nu_{ij})$. We wish to use this connection to reconstruct $\E_\G$ and $N$ by running the network with exogenous inputs \textcolor{black}{$\mathrm{w^{(1)},\ldots,w^{(n)}}$} and measuring the corresponding steady-state outputs \textcolor{black}{$\mathrm{y^{(1)},\ldots,y^{(n)}}$}.

\section{A Network Identification Algorithm \\For Convergent Networks} \label{sec.GeneralAlgorithm}
This section presents an algorithm solving Problem \ref{prob.network_detection2}. We consider equation \eqref{eq.ClosedLoopEquation}, relating a constant exogenous input $\mathrm w$ and the resulting steady-state output $\mathrm y$. If the functions $k^{-1},g$ were linear, the relation \eqref{eq.ClosedLoopEquation} is also linear, meaning that $\E_\G N\E_\G^\top$ can be found by taking $n$ steady-state input-output pairs $(\textcolor{black}{\mathrm w^{(i)}},\textcolor{black}{\mathrm y^{(i)}})$, assuming the inputs \textcolor{black}{
$\mathrm w^{(1)},\ldots,\mathrm w^{(n)}$} are linearly independent. Alas, the functions $k^{-1},g$ need not be linear.  However, we can still apply the same idea by \emph{linearizing} equation \eqref{eq.ClosedLoopEquation}. For now, we assume the steady-state inputs and outputs can be measured. We explain how to deal with this assumption at the end of this section.

We first run the network with an arbitrary constant exogenous input \textcolor{black}{$\mathrm w^{(0)}$}, get the corresponding steady-state output \textcolor{black}{$\mathrm y^{(0)}$}, and linearize \eqref{eq.ClosedLoopEquation} around \textcolor{black}{$\mathrm y^{(0)}$}. Thus, for a constant exogenous input $\mathrm w = \textcolor{black}{\mathrm w^{(0)}} + \delta \mathrm w$, we obtain
\textcolor{black}{
\begin{align} \label{eq.Nonlinear}
&\mathrm w - k^{-1}(\mathrm y) = \E_\G N g(\E_\G^\top \mathrm y)\approx \\ &\E_\G N g(\E_\G^\top \mathrm y^{(0)}) + \E_\G N \nabla g(\E_\G^\top \mathrm y^{(0)}) \E_\G^\top \delta \mathrm y, \nonumber
\end{align}}
where $\mathrm y$ is the steady-state output of the network, and $\delta \mathrm y = \mathrm{y - \textcolor{black}{y^{(0)}}}$. More precisely, we have the following result.

\begin{prop} \label{prop.QuadraticOutputError}
Suppose that the functions $k^{-1},g$ are twice differentiable at $\textcolor{black}{\mathrm{y}^{(0)}}$ and $\E_\G^\top\textcolor{black}{\mathrm{y}^{(0)}}$ respectively. For any $\delta \mathrm w$ small enough, and $\mathrm y = \textcolor{black}{\mathrm y^{(0)}} + \delta \mathrm y$, the following equation holds: %, where $\mathrm y = \mathrm y_0 + \delta \mathrm y$.
\begin{align} \label{eq.9}
\left[\nabla k^{-1}(\textcolor{black}{\mathrm y^{(0)}}) + \E_\G N \nabla g(\E_\G^\top \textcolor{black}{\mathrm y^{(0)}}) \E_\G^\top\right] \delta \mathrm y  = \delta \mathrm w + O(\|\delta \mathrm y\|^2).
\end{align} 
\end{prop}

\begin{proof}
By subtracting $\textcolor{black}{\mathrm w^{(0)}} = k^{-1}(\textcolor{black}{\mathrm y^{(0)}}) + \E_\G N g(\E_\G^\top \textcolor{black}{\mathrm y^{(0)}})$ from $\mathrm w= k^{-1}(\mathrm y) + \E_\G N g(\E_\G^\top \mathrm y)$, we get:
\begin{align*}
\delta \mathrm w = k^{-1}(\mathrm y) - k^{-1}(\textcolor{black}{\mathrm y^{(0)}}) + \E_\G N \left(g(\E_\G^\top\mathrm y) - g(\E_\G^\top \textcolor{black}{\mathrm y^{(0)}})\right).
\end{align*}
The result now follows by using $g$'s and $k^{-1}$'s Taylor expansions near $\E_\G^\top \textcolor{black}{\mathrm y^{(0)}}$ and $\textcolor{black}{\mathrm y^{(0)}}$ up to first order, where the error term is quadratic as $g$ and $k^{-1}$ are twice differentiable.
\end{proof}

We denote $\mathcal{M}=\nabla k^{-1}(\textcolor{black}{\mathrm y^{(0)}}) + \E_\G N \nabla g(\E^\top_\G \textcolor{black}{\mathrm y^{(0)}}) \E_\G^\top$. Proposition \ref{prop.QuadraticOutputError} suggests a way to estimate the matrix $N$ and the graph $\G$. Indeed, we inject $n$ constant exogenous inputs and measure the corresponding steady-state outputs, yielding \textcolor{black}{$\delta\mathrm w^{(1)},\delta\mathrm y^{(1)},\ldots,\delta\mathrm w^{(n)},\delta \mathrm y^{(n)}$} as in the proposition. \eqref{eq.9} implies that $\mathcal M \delta \textcolor{black}{\mathrm y^{(i)}} \approx \delta \textcolor{black}{\mathrm w^{(i)}}$ for $i=1,2,\ldots,n$. If \textcolor{black}{$\delta \mathrm y^{(1)},\ldots,\delta\mathrm y^{(n)}$} are linearly independent, then we have $\mathcal M \approx \delta W\delta Y^{-1}$ where \textcolor{black}{$\delta W = \left[\begin{smallmatrix} \delta \mathrm w^{(1)} & \cdots & \delta \mathrm w^{(n)}\end{smallmatrix}\right]$} and \textcolor{black}{$\delta Y = \left[\begin{smallmatrix} \delta \mathrm y^{(1)} & \cdots & \delta \mathrm y^{(n)}\end{smallmatrix}\right]$}. In turn, we can estimate $N$ and the graph $\G$ by looking at the off-diagonal entries of $\delta W\delta Y^{-1}$. Thus, we design \textcolor{black}{$\mathrm w^{(0)},\mathrm w^{(1)},\ldots,\mathrm w^{(n)}$} in a way assuring that \textcolor{black}{$\delta \mathrm y^{(1)},\ldots \delta \mathrm y^{(n)}$} are linearly independent. We also note that $\mathcal{M}\mathbbm{1}_n = \nabla k^{-1}(\textcolor{black}{\mathrm y^{(0)})} \mathbbm{1}_n$ is independent of the unknown matrix $N$ and the graph $\G$, so we can choose one $\delta \textcolor{black}{\mathrm y^{(i)}}$ as $\mathbbm{1}_n$, which will be useful in one specific case. 

Namely, if $k^{-1}$ is a constant function, solutions to \eqref{eq.ClosedLoopEquation} are unique up to a multiple of $\mathbbm{1}_n$, namely, different steady-states might be achieved from different initial conditions. For example, consider the consensus protocol, concerning MANs with single integrators agents and controllers given by $g_{ij}(\zeta_{ij}) = \zeta_{ij}$. In this case, $k^{-1}(\mathrm y) = 0$ for any $\mathrm y$, and the agents converge to consensus whose value depends on initial conditions. For that reason, if $k^{-1}$ is constant near \textcolor{black}{$\mathrm y^{(0)}$}, we project \textcolor{black}{$\delta \mathrm y^{(1)},\ldots,\delta \mathrm y^{(n)}$} to $\mathbbm{1}_n^\perp$, the space orthogonal to $\mathbbm{1}_n$. 

\begin{thm} \label{thm.GeneralIndependence}
Let $\mathbb{P}_{n,0}$ be an absolutely continuous probability measure on $\mathbb{R}^n$, and take \textcolor{black}{$\mathrm w^{(0)}$} as a sample of $\mathbb{P}_{n,0}$. Let \textcolor{black}{$\mathrm y^{(0)}$} be a corresponding steady-state output, i.e. a solution to \eqref{eq.ClosedLoopEquation} with input $\mathrm w^{(0)}$. For some $\kappa \in \R$, define \textcolor{black}{$\{\delta \mathrm y^{(i)}\}_{i=1}^n$} as follows:
\begin{itemize}
\item Suppose that $k^{-1}$ is differentiable at \textcolor{black}{$\mathrm y^{(0)}$} and that $\nabla k^{-1}(\textcolor{black}{\mathrm y^{(0)}}) = 0$. For any $i=1,\ldots,n-1$, choose $\delta \textcolor{black}{\mathrm w^{(i)}}= \kappa (\mathrm e_i - \mathrm e_n)$, take \textcolor{black}{$\mathrm{y^{(i)}}$} as the steady-state output corresponding to \textcolor{black}{$\mathrm w^{(i)}$}, and define \textcolor{black}{$\delta \mathrm y^{(i)} = (\mathrm{Id}_n - \frac{1}{n}\mathbbm{1}_n\mathbbm{1}_n^\top)(\mathrm y^{(i)} - \mathrm y^{(0)})$}. Also, set $\delta\textcolor{black}{\mathrm y^{(n)}} = \kappa\mathbbm{1}_n$.
\item Otherwise, choose $\delta\textcolor{black}{\mathrm w^{(i)}} = \kappa \mathrm e_i$ for $i=1,\ldots,n$. Define \textcolor{black}{$\mathrm{y^{(i)}}$} as the steady-state output corresponding to $\textcolor{black}{\mathrm w^{(i)}}$ and \textcolor{black}{$\delta \mathrm y^{(i)} = \mathrm y^{(i)} - \mathrm y^{(0)}$}. 
\end{itemize}
Suppose Assumptions \ref{Assumption1} and \ref{Assumption2} hold. If $\kappa$ is small enough, then the set $\mathcal{A}= \textcolor{black}{\{\delta \mathrm y^{(1)},\ldots,\delta\mathrm y^{(n-1)},\delta \mathrm y^{(n)}\}}$ is a basis for $\mathbb{R}^n$.
\end{thm}

The proof of Theorem \ref{thm.GeneralIndependence} can be found in Appendix \ref{append.ProofIndependence}. The resulting algorithm, as described in the paragraphs proceeding Theorem \ref{thm.GeneralIndependence}, is summarized as Algorithm \ref{alg.General}, with minor changes. 
Namely, For reasons explained later (see Remark \ref{rem.Omega1Implementation2}), it would be advantageous if \textcolor{black}{$\delta \mathrm w^{(1)},\ldots\delta \mathrm w^{(n)}$} are also linearly linearly independent. This is problematic if $\nabla k^{-1}(\textcolor{black}{\mathrm y^{(0)}}) = 0$, as we take that $\delta \textcolor{black}{\mathrm y^{(n)}} = \kappa\mathbbm{1}_n$, and compute $\mathrm w^{(n)}$ as $\mathrm w^{(n)} = \mathcal{M}\delta y^{(n)} = \kappa \nabla k^{-1}(\textcolor{black}{\mathrm y^{(0)}}) \mathbbm{1}_n = 0$. Instead, we declare $\delta \textcolor{black}{\mathrm w^{(n)}} = \kappa\mathbbm{1}_n$, so that $\delta W\delta Y^{-1} \approx \mathcal M + \frac{1}{n}\mathbbm{1}_n\mathbbm{1}_n^\top$, yielding an estimate $\mathcal M$ as $\delta W\delta Y^{-1} - \frac{1}{n}\mathbbm{1}_n\mathbbm{1}_n^\top$.

\begin{algorithm} 
\caption{Network Identification for Convergent Networks}
{\bf Input:} A networked system on $n$ agents with an exogenous input $w(t)$ and a measurable output $y(t)$.

{\bf Output:} An estimate of the underlying graph, $\textcolor{black}{\hat{\mathcal{G}}}$, and an estimate of the edge weights, $\{\textcolor{black}{\hat{\nu}_{ij}}\}$. 
\begin{algorithmic}[1] \label{alg.General}
\STATE{Randomly choose \textcolor{black}{$\mathrm w^{(0)}$} as a standard Gaussian vector. Change the value of the exogenous input $w(t)$ to \textcolor{black}{$\mathrm w^{(0)}$}}.
\STATE{Wait for the diffusively coupled network to converge. Measure its steady-state output and denote it as $\textcolor{black}{\mathrm y^{(0)}}$.}
\STATE Choose $0 < \kappa \ll 1$.
\IF{$\nabla k^{-1}(\textcolor{black}{\mathrm y^{(0)}})=0$}
\STATE Define $\delta \textcolor{black}{\mathrm w^{(i)}} = \kappa(\mathrm e_i-\mathrm e_n)$ for $i=1,\ldots,n-1$.
\STATE Put $\delta\textcolor{black}{\mathrm y^{(n)}} = \kappa \mathbbm{1}_n$ and $\delta \textcolor{black}{\mathrm w^{(n)}} = \kappa\mathbbm{1}_n$.
\STATE Put $\text{NumRuns} = n-1$,  $Q = \frac{1}{n}\mathbbm{1}_n\mathbbm{1}_n^\top$, and $J = \mathrm{Id}_n - Q$.
\ELSE
\STATE Define $\delta \textcolor{black}{\mathrm w^{(i)}} = \kappa \mathrm e_i$ for $i=1,\ldots,n$.
\STATE Put $\text{NumRuns} = n$, $Q = 0$, and $J = \mathrm{Id}_n$.
\ENDIF
\FOR{$i=1$ to $\text{NumRuns}$}
\STATE Change the value of $w(t)$ to \textcolor{black}{$\mathrm w^{(0)} + \delta\mathrm w^{(i)}$}.
\STATE Wait for the diffusively coupled network to converge. Measure its steady-state output and denote it as $\textcolor{black}{\mathrm y^{(i)}}$.
\STATE Define \textcolor{black}{$\delta \mathrm y^{(i)} = J(\mathrm y^{(i)} - \mathrm y^{(0)})$}.
\ENDFOR
\STATE Define \textcolor{black}{$\delta W = \left[\begin{smallmatrix} \delta \mathrm w^{(1)} & \ldots & \delta \mathrm w^{(n)}\end{smallmatrix}\right]$} and \textcolor{black}{$\delta Y = \left[\begin{smallmatrix} \delta \mathrm y^{(1)} & \ldots & \delta \mathrm y^{(n)}\end{smallmatrix}\right]$}.
\STATE Compute $M^\prime = \delta W \delta Y^{-1}$ and $M = M^\prime - Q$..
\STATE Define an empty graph \textcolor{black}{$\hat{\mathcal{G}}$} on $n$ nodes. 
\FOR{$i=1$ to $n$ and $j=1$ to $n$}
\IF{$|M_{i,j}| > \varepsilon$ and $i\neq j$}
\STATE Add the edge $\{i,j\}$ to the graph $\textcolor{black}{\hat{\mathcal{G}}}$.
\STATE Define $d_{ij} = \frac{dg_{ij}}{d\zeta_{ij}}((\textcolor{black}{\mathrm y^{(0)})_i - (\mathrm y^{(0)})_j})$; $\textcolor{black}{\hat{\nu}_{ij}} = -\frac{M_{ij}}{d_{ij}}$.
\ENDIF
\ENDFOR
\STATE {\bf Return} the graph \textcolor{black}{$\hat{\mathcal{G}}$} and the coupling coefficients $\{\textcolor{black}{\hat{\nu}_{ij}}\}$.
\end{algorithmic}
\end{algorithm}

Algorithm \ref{alg.General} gives an identification scheme - choose a collection of linearly independent vectors, run the MAN using them as inputs, measure the steady-state outputs, and use the result to compute estimates of the graph $\hat{\G}$ and the weights $\hat{\nu}_{ij}$.
Instead of running the MAN multiple times, we can apply a switching signal and use the global convergence of the MAN, i.e., we inject an exogenous input $w(t)$ whose value changes when the MAN reaches a steady-state, or $\epsilon$-close to it. 

It is clear that unless all agents and controllers \textcolor{black}{are} LTI, Algorithm 1 is an approximation algorithm, as the quadratic error term in \eqref{eq.Nonlinear} affects the output. We now bound the error of algorithm and determine its time complexity.

\begin{thm} \label{thm.GeneralErrorRateAndComplexity}
Consider a network $(\G_\nu,\Sigma,g)$ satisfying assumptions \ref{Assumption1} and \ref{Assumption2}.
\begin{enumerate}
\item Let $M$ be the matrix calculated by Algorithm \ref{alg.General}. Then for any $i,j\in\V$, we have $$|M_{ij}-\mathcal M_{ij}| \le O\left(\sqrt{n}\kappa\left(1 + \max_{i,j}(\nu_{ij}d_{ij})\lambda_{\text{max}}(\G)\right)\right),$$ with probability 1. Thus, Algorithm \ref{alg.General} approximates the graph and coupling coefficients, with probability $1$. 
\item \textcolor{black}{Algorithm \ref{alg.General} performs $O(n^{\omega})$ floating point operations.}
\end{enumerate}
\end{thm}

The proof of Theorem \ref{thm.GeneralErrorRateAndComplexity} can be found in Appendix \ref{append.ProofAlgorithm}. Its main idea is to bound the operator norm of $\|M - \mathcal M\|$ by $\|(M-\mathcal M)\delta Y\| \|\delta Y^{-1}\|$. The first part is estimated using \eqref{eq.9}, and the latter is estimated by using the matrices $\mathcal M$ and $\delta W$. 

\begin{rem}
If the agents and controllers are LTI, equation \eqref{eq.ClosedLoopEquation} is already linear. Thus, Algorithm \ref{alg.General} is error-less.
\end{rem}

\begin{rem}
Steps 17-18 of Algorithm \ref{alg.General} essentially solve the least-squares problem \textcolor{black}{$\sum_{i=1}^N \|M^\prime \delta \mathrm y^{(i)} - \delta \mathrm w^{(i)}\|^2$}. If, for some reason, we decide to take more than $n$ measurements, we could reformulate steps 17-18 with a least-squares problem.
\end{rem}

\begin{rem}
\textcolor{black}{
Algorithm \ref{alg.General} does not require full knowledge of the steady-state relations, but only if $\nabla k^{-1}(\mathrm y^{(0)}) = 0$, as well as the derivative of $g_{ij}$ at a single point. For the former, note that $\nabla k^{-1}$ is a diagonal matrix whose entries are $\nabla k_i^{-1}$, so $\nabla k^{-1}(\mathrm y^{(0)}) = 0$ if and only if $\nabla k^{-1}(\mathrm y^{(0)})\mathbbm{1}_n = 0$, which can be tested by injecting the constant exogenous input $\mathrm w = \mathrm w^{(0)} + \kappa \mathbbm{1}_n$.}
\end{rem}

\begin{rem} \label{rem.Omega1Implementation2}
We can reduce the complexity of computing $M^\prime$ from $O(n^\omega)$ to $O(n^{\omega_1})$, reducing the number of floating point operations to $O(n^{\omega_1})$. Indeed, we have $(M^\prime)^{-1} = \delta W^{-1} \delta Y$. As $\delta W$ is the product of $O(n)$ elementary matrices (see Lemma \ref{lem.InvW}), we can compute $(M^\prime)^{-1}$ in $O(n^2)$ time by applying the corresponding row operations on $\delta Y$. However, the matrix $(M^\prime)^{-1}$ need not be positive-definite, or even symmetric, due to the error term in \eqref{eq.12}. We thus consider the symmetric matrix $(M^\prime)^{-1}_{\text{sym}} = \frac{1}{2}((M^\prime)^{-1}+((M^\prime)^{-1})^\top)$. This matrix is close to the inverse of $\mathcal{M} + Q > 0$, hence, the eigenvalues of $(M^\prime)_{\text{sym}}^{-1}$ are positive and it is therefore a positive-definite matrix, so inverting it costs only $O(n^{\omega_1})$ time. 
\end{rem}

Theorem \ref{thm.GeneralErrorRateAndComplexity} gives an error bound on the elements of the matrix $M$, but we want a more explicit error estimate on the weighted graph $\textcolor{black}{\hat{\mathcal{G}}_{\hat \nu}}$ computed by the algorithm. We now relate the estimate on $|M_{ij} - \mathcal M_{ij}|$, the estimates on the weights $\{\nu_{ij}\}$, and the estimate on the graph $\G$.
 
\begin{prop} \label{prop.GraphRigidityLemma}
Suppose the same assumptions as in Theorem \ref{thm.GeneralErrorRateAndComplexity} hold. Suppose further that for any $i,j \in \V$, $|M_{ij} - \mathcal M_{ij}| \le m$, and that $m\le  \frac{1}{4}\min_{i,j} (\nu_{ij}d_{ij})$ . If the number $\varepsilon$ used in Algorithm \ref{alg.General} is equal to $2m$, then the graph $\textcolor{black}{\hat{\mathcal{G}}}$ computed by the algorithm satisfies $\textcolor{black}{\hat{\mathcal{G}}} = \G$. Moreover, for all $\{i,j\}\in\EE$, the inequality $|\textcolor{black}{\hat{\nu}_{ij}} - \nu_{ij}| \le  d_{ij}^{-1}m$ holds.
\end{prop}

\begin{proof}
Suppose first that $\{i,j\}\not\in\EE$. Then $\mathcal M_{ij} = 0$, meaning that $|M_{ij}| \le m < \varepsilon$. Thus $\{i,j\}$ is not in $\textcolor{black}{\hat{\mathcal{G}}}$, as required. Conversely, assume that $\{i,j\}\in\EE$. Then $\mathcal M_{ij} = -d_{ij}\nu_{ij}$, and $|M_{ij}| \ge d_{ij}\nu_{ij} - m \ge d_{ij}\nu_{ij} - \frac{1}{4}d_{ij}\nu_{ij}  > \varepsilon$. Thus, $\{i,j\}\in \textcolor{black}{\hat{\mathcal{G}}}$. Moreover, $M_{ij} = -\textcolor{black}{\hat{\nu}_{ij}}d_{ij}$ and $\mathcal M_{ij} = -\nu_{ij}d_{ij}$, together with $|M_{ij} - \mathcal M_{ij}| \le m$, imply that $|\textcolor{black}{\hat{\nu}_{ij}} - \nu_{ij}| \le  d_{ij}^{-1}m$.
\end{proof}

Theorem \ref{thm.GeneralErrorRateAndComplexity} and Proposition \ref{prop.GraphRigidityLemma} show that Algorithm \ref{alg.General} approximates the underlying weighted graph. \textcolor{black}{Moreover, they show that it performs $O(n^{\omega})$ floating point operations, where Remark \ref{rem.Omega1Implementation2} shows it can be reduced to $O(n^{\omega_1})$. %\footnote{Recall from Section \ref{subsec.passivity_complexity} that $2\leq \omega_1\leq\omega<3 $.}
However, we do not consider the time it takes the network to converge, corresponding to steps 2 and 13-14 in Algorithm 1. Indeed, }as steady-states are only achieved asymptotically, we must use a finite-time approximation, \textcolor{black}{ideally after no more than $O(n^{\omega -1})$ time, so that Algorithm \ref{alg.General} will have a runtime of $O(n^\omega)$. As the value of $\omega \ge 2$ isn't actually known, we use the more conservative bound of $O(n)$ instead.}

\textcolor{black}{Section \ref{sec.Discussion} below, which discusses the robustness of Algorithm \ref{alg.General}, also gives an estimate on its error as a function of measurement inaccuracies. Namely, Proposition \ref{prop.MeasError} below shows that we should measure the output when its distance from the true steady-state is bounded by $O(n^{-0.5})$. In order to get this close to the steady-state in $O(n)$ time, we need the network to have a convergence rate of no more than $O(t^{-0.5})$. This is certainly achieved by LTI networks satisfying Assumption \ref{Assumption1}, which have an exponential rate of convergence. This is also achieved for networks with output-strictly MEIP agents and MEIP controllers, as the ``convergence profile" method described in \cite[Section VI]{Sharf2019f} shows that under minor technical assumptions, the network converges no slower than the solution of the ODE $\dot{s} = -Cs^{\beta}$ where $0<\beta<2$ depends on the observation function of the agents and $C>0$ depends on passivity coefficients. The case $\beta < 1$ gives finite-time convergence, $\beta = 1$ gives exponential convergence, and $1<\beta<2$ gives polynomial convergence at a rate of $O(t^{-1/(\beta-1)}) < O(t^{-0.5})$. In other cases, however, the desired convergence rate can only be achieved by strengthening Assumption 1 and demanding it explicitly, i.e. that the network has a convergence rate no slower than $O(t^{-0.5})$. In all three cases, the algorithm has a total runtime of $O(n^\omega)$. }

As for determining when to declare that the network converged to a steady-state, there are many ways do so. For networks satisfying the same MEIP-based condition as above, the ``convergence profile" method described in \cite[Section VI]{Sharf2019f} can be used to determine a stopping time, namely when the solution of $\dot{s} = -Cs^{\beta}$ has become small enough. Another solution is to stop running the network when $\dot y$ (or $\dot x$) is small enough. Other ways to determine the stopping time include computer-based simulations, or even intuition based on the physical constants affecting the agents' dynamics.

\vspace{-5pt}
\section{Robustness and Practical Considerations}\label{sec.Discussion}
Algorithm \ref{alg.General} solves Problem \ref{prob.network_detection2}, but does so under some strong assumptions. First, it assumes the dynamics are noiseless and disturbance-free, allowing it to converge to a constant steady-state. Second, it assumes the measurements taken are perfect and are not subject to disturbances. Lastly, it assumes the exogenous input can be applied to all agents. This section is dedicated to discuss these points, and to give a briefly compare the algorithm to other methods described in literature.

\vspace{-7pt}
\subsection{Robustness to Noise and Disturbances}
We begin by studying how noise and disturbances affect the output of the diffusively-coupled network. Generally, if we make no passivity assumption on the network, then it might not converge in the presence of noise. One example of this phenomenon is the consensus protocol \cite{OlfatiSaber2007}, for which white noise does not disturb the asymptotic convergence to consensus (almost surely), but it does cause the consensus value to be volatile. The consensus protocol can also be viewed as the MAN with single-integrator agents and static gain controllers, meaning it has passive agents and output-strictly passive controllers. However, passivity can still be used for some form of noise- and disturbance-rejection:
\begin{prop} \label{prop.PassivityNoise}
Consider a diffusively-coupled network $(\G_\nu,\Sigma,g)$ with steady-state $(\mathrm{u,y},\upzeta,\upmu)$. Suppose that the agents are output-strictly passive with respect to $(\mathrm{u}_i,\mathrm y_i)$ with parameters $\rho_i > 0$, and that the controllers are passive with respect to $(\upzeta_e,\upmu_e)$. Let $S$ be the sum of the agents' storage functions, and denote $R = \mathrm{diag}(\rho_i) > 0$. Let $0<\Delta \in \R$.

Consider a parasitic exogenous input $d(t)$ to the agents, so the input is $u(t) = d(t) - \E_\G \mu(t)$, and assume that at any time, $||R^{-1/2}d(t)|| \le \Delta$ holds. For any $\varsigma > 0$, define $\mathcal{A}_\varsigma = \{x:\ \|R^{1/2}(h(x)-\mathrm y)\| \le \varsigma\}$, and let $\Xi_\varsigma = \max_{x\in \mathcal{A}_\varsigma} S(x)$. Then for any $\epsilon > 0$ and any initial condition, there exists some $T$ such that if $t>T$, then $\|y(t)-\mathrm y\|\le \max_{x:\ S(x)\le \Xi_{\Delta+\epsilon}}{\|h(x)-\mathrm y\|}$. 
\end{prop}
\begin{proof}
Let $v(t)=-\E_\G \mu(t)$. By \textcolor{black}{definition of (output-strict)} passivity, we have that $0 \le \sum_{e\in\EE} (\zeta_e - \upzeta_e)(\mu_e - \upmu_e)$, and
\begin{align*}
\frac{d}{dt} S(x)\le& \sum_{i\in \V}\left(-\rho_i \|y_i - \mathrm y_i\|^2 + (d_i+v_i-\mathrm{u}_i)(y_i - \mathrm y_i)\right),
\end{align*}
as $v(t) = u(t)-d(t)$. Summing the equations and using $v = -\E_\G \mu$, $\zeta = \E_\G^\top y$, $\mathrm u = -\E_\G \upmu$ and $\upzeta = \E_\G^\top \mathrm y$ yields,
\begin{align} \label{eq.PassivityNoiseIneq}
\frac{d}{dt} S(x) &= -(y-\mathrm y)^\top R(y-\mathrm y) + d(t)^\top (y-\mathrm y) \nonumber \\
&= -||R^{1/2}(y-\mathrm y)||^2 + (R^{-1/2}d(t))^\top R^{1/2}(y-\mathrm y) \nonumber\\
&\leq -||R^{1/2}(y-\mathrm y)||^2 + \Delta\|R^{1/2}(y-\mathrm y)\|.
\end{align}
We note that if $\|R^{1/2}(y(t)-\mathrm y)\| > \Delta+\epsilon$, then the right-hand side is bounded from above by $-(\epsilon+\Delta)\epsilon < 0$. If this happens indefinitely, we eventually reach $S(x) < 0$, which is absurd. Thus, for some $T>0$, we have $\|R^{1/2}(y(T)-\mathrm y)\| \le \Delta+\epsilon$, hence $x(T)\in \mathcal{A}_{\Delta+\epsilon}$ and $S(x(T))\le \Xi_{\Delta+\epsilon}$. 

We now claim that $S(x(t)) \le \Xi_{\Delta+\epsilon}$ for all $t>T$. If not, we find a time $t_1$ in which $S(x(t_1)) > \Xi_{\Delta+\epsilon}$.
By the same argument, we find some time $T^{\prime} > t_1$ such that $S(x(T^{\prime})) \le \Xi_{\Delta + \epsilon}$. The function $S(x(t))$ is continuous for times $T\le t \le T^{\prime}$, so it attains a maximum at some time $t_2 > T$. As $S(x(t_1))$ is larger than $S(x(T)),S(x(T^\prime))$, the point $t_2$ must be in the interior of the interval between $T$ and $T^\prime$, which implies that $\frac{d}{dt}S(x(t)) = 0$ at $t=t_2$. However, $S(x(t_2)) \ge S(x(t_1)) >  \Xi_{\Delta+\epsilon}$, so by definition of $\Xi_{\Delta+\epsilon}$, we must have $x(t_2)\not\in \mathcal{A}_{\Delta+\epsilon}$, i.e. $\|R^{1/2}(y(t_2)-\mathrm y)\| > \Delta+\epsilon$. In turn, this implies that $\frac{d}{dt}S(x(t)) < 0$ at $t=t_2$ by \eqref{eq.PassivityNoiseIneq}. This is a contradicition, as we saw that $\frac{d}{dt}S(x(t)) = 0$ at $t=t_2$. Thus $S(x(t)) \le \Xi_{\Delta+\epsilon}$ for all $t\ge T$. This completes the proof.
\end{proof}
The proposition above shows that even in the presence of disturbances or noise, the algorithm can sample the output $y(t)$ not too far from the true, disturbance-free steady-state output $\mathrm y$. This result will intertwine with Proposition \ref{prop.MeasError}, in which the effects of measurement errors will be accounted for.
\begin{rem}
Proposition \ref{prop.PassivityNoise} does not distinguish between disturbances and random noise. In practice, the bound is of the right order of magnitude for disturbances, but is a gross overestimate for noise. For example, consider a agents $\dot{x} = -x_i+u_i, y_i=x_i$ with weights $\nu_{ij} = 0$, where $d(t)$ is chosen as random white noise, bounded by $C$, and with variance $\upsigma^2$. Proposition \ref{prop.PassivityNoise} shows the agents converges to an output with $|y(t)|\le C$ (as $\rho = 1$). However, writing $x(t)$ as a convolution integral and applying It\^{o} calculus \cite{Oksendal2013} shows that $\EE(x(t)) = 0$ and that $\mathrm{Var}(x(t)) \le \upsigma^2/2$. Chebyshev's inequality now gives a high-probability bound on where $x(t)$ can be, which is much better than Proposition \ref{prop.PassivityNoise}'s bound if $\upsigma \ll C$. 
\end{rem}

We now shift our focus to measurement errors, which add parasitic terms when defining the matrices $\delta W$, $\delta Y$ in Algorithm \ref{alg.General}. We prove the following:
\begin{prop}\label{prop.MeasError}
Suppose Algorithm \ref{alg.General} builds the matrix $\delta Y + \Delta Y$ instead of $\delta Y$, due to measurement error. If $||\Delta Y|| \ll ||\delta Y||$, then the algorithm calculates a matrix $M$ with $\|M-\mathcal M\| \le O\left(\sqrt{n}\left(1+\max_{i,j}(\nu_{ij}d_{ij})\lambda_{\mathrm{max}}(\G)\right)\right)\|\Delta Y \delta Y^{-1}\|$, plus the error term from Theorem \ref{thm.GeneralErrorRateAndComplexity}.
\end{prop}
The proof of the proposition can be found in Appendix \ref{append.ProofAlgorithm}.
\begin{rem}
Proposition \ref{prop.MeasError} gives bounds the error if a bound on the relative measurement error is known. In some cases, we have a bound on the absolute measurement error, e.g. Proposition \ref{prop.PassivityNoise}. In that case, we have $\|\Delta Y \delta Y^{-1}\| \le \|\delta Y^{-1}\| \|\Delta Y\|$, where $\|\delta Y^{-1}\|$ can be bounded as in the proof of Theorem \ref{thm.GeneralErrorRateAndComplexity}, i.e. $\|\delta Y^{-1}\| \le O(\kappa^{-1}(1+\max_{i,j}(\nu_{ij}d_{ij})\lambda_{\mathrm{max}}(\G)))$.
\end{rem}

\vspace{-15pt}
\subsection{Probing Inputs Supported on Subsets of Nodes}
The previous subsection shows that the algorithm is somewhat resistant to noise, either in the dynamics or the measurement.
We now move to the last major assumption, namely that the exogenous input can be applied to all agents. There are two possible ways to relax the assumption. First, we can apply compressed sensing methods, using the sparsity of $M^\prime$, which corresponds to a a sparse graph $\G$, similarly to \cite{Gardner2003,Guangjun2015}. See \cite{Donoho2006} for more on compressed sensing.
Another approach is to still try and use $n$ different measurements, with exogenous inputs supported only on $\ell$ nodes. Indeed, in order to reconstruct the matrix $M^{\prime}$, the vectors \textcolor{black}{$\delta \mathrm y^{(1)},\ldots,\delta \mathrm y^{(n)}$} must span $\R^n$. If the steady-state equation \eqref{eq.ClosedLoopEquation} is inherently nonlinear, then even when the inputs are restricted to a subspace of dimension $\ell$, the outputs can span all of $\mathbb{R}^n$. Abstractly, we prove the following:
\begin{prop} \label{prop.RankMap}
Let $F:\R^d \to \R^n$ be any function which is $(\ell+1)$-times differentiable at \textcolor{black}{$\mathrm w^{(0)} \in \R^d$}. Suppose the dimension of the subspace spanned by all partial derivatives of $F$ at \textcolor{black}{$\mathrm w^{(0)}$} up to order $\ell$ is $r$, and denote the number of these partial derivatives by $s$. Let $\PP$ be any absolutely continuous probability measure on $\R^d$ which is supported on a small ball around \textcolor{black}{$\mathrm w^{(0)}$}, and let \textcolor{black}{$\mathrm w^{(1)},\ldots,\mathrm w^{(s)}$} be i.i.d. samples of it. Then, with probability 1, the span of the vectors \textcolor{black}{$F(\mathrm w^{(1)})-F(\mathrm w^{(0)}),\ldots, F(\mathrm w^{(s)})-F(\mathrm w^{(0)})$} has dimension $r$. In particular, if $r=n$ then they span all of $\R^n$.   
\end{prop}
Before proving the theorem, we emphasize that the harsher smoothness assumptions are made on the steady-state relation $k^{-1}$ and the interaction $g$, and that the agents' dynamics might still be non-smooth. We also emphasize that these stricter assumptions are only used for the proof, and Algorithm \ref{alg.General} would still only use the first derivative of $k^{-1}$ and $g$. %We now prove the theorem:
\begin{proof}
Suppose that $\PP$ is supported inside a ball around $\textcolor{black}{\mathrm{w}^{(0)}}$ of radius $\kappa \ll 1$. By Taylor's theorem, we can write $\Phi_F = D_F W$ up to an error of order $O(\kappa^{\ell+1})$, where $\Phi_F$ is the matrix whose columns are \textcolor{black}{$F(\mathrm w^{(i)})-F(\mathrm w^{(0)})$} , $D_F$ is the matrix whose columns are all the partial derivatives of $F$ at \textcolor{black}{$\mathrm w^{(0)}$} up to order $\ell$, and $W$ is a square matrix consisting of monomials of the entries of \textcolor{black}{$\mathrm w^{(i)} - \mathrm w^{(0)}$} for $i\in\{1,\ldots,s\}$. We assumed that ${\rm rank}~D_F = r$, and we aim to prove that ${\rm rank}~\Phi_F = r$. This immediately follows if $W$ is invertible, which we now prove.
Indeed, consider the map $p:(\R^d)^s\to\R^n$ defined by $p(\textcolor{black}{\mathrm w^{(1)},\ldots,\mathrm w^{(s)}}) = \det W$. This is a non-zero polynomial in \textcolor{black}{$\mathrm w^{(1)},\ldots,\mathrm w^{(s)}$}, and $W$ is invertible if and only if $p \neq 0$. However, the collection of zeros of a non-zero polynomial is a zero-measure set \cite{Caron2005}, and thus $p\neq 0$ with probability $1$, so $W$ is invertible with probability 1. This concludes the proof.
\end{proof}
Proposition \ref{prop.RankMap} can be applied to the map $F$ mapping $\mathrm w$ to $\mathrm y$ according to \eqref{eq.ClosedLoopEquation}. In some occasions, it can be hard to compute the rank $r$, but one can use the proposition in a more data-driven fashion - take $s$ random samples \textcolor{black}{$\mathrm w^{(0)} + \delta \mathrm w^{(i)}$} near \textcolor{black}{$\mathrm w^{(0)}$}, and compute the $s$ corresponding steady-state outputs \textcolor{black}{$\mathrm y^{(0)} + \delta \mathrm y^{(i)}$}. The rank $r$ is computed using $\delta \textcolor{black}{\mathrm y^{(i)}}$, and one can find the connecting matrix $M^{\prime}$ using compressed sensing. 

To conclude this section, we saw that the presented algorithm can be applied in real-world scenarios, in which noise and measuring errors exist, and not all nodes are susceptible to controlled exogenous inputs. Other algorithms which use probing inputs, or similar methods, rely on linearizing the dynamics instead of the steady-state equation, using higher-order terms in the Taylor approximation, or assuming that the dynamics are Lipschitz continuous \cite{Chen2009,Gardner2003,Guangjun2015,Prasse2020,Burbano2019}. Such methods cannot be applied if the dynamics are non-smooth, or even discontinuous, e.g., when dry friction is introduced \cite{vanderLinden1993}, or for some finite-time consensus protocols \cite{Shi2018}. Moreover, such methods are usually applied by measuring the network at regularly scheduled intervals, which yields nearly identical measurements if the dynamics are either very slow or that convergence happens extremely quickly, thus wasting power on unnecessary sensing and communication. On the contrary, these methods work well for networks with relatively rich dynamics, where useful measurements are gathered regularly.

In comparison, the presented algorithm can still be applied in the case of non-smooth or even discontinuous dynamics, so long that the steady-state relations are differentiable. It will not function as well in dynamics-rich networks, does not waste unnecessary power for networks with slow dynamics, and will perform superbly for networks which converge fast. Examples of such networks include networks of autonomous vehicles trying to coordinate their velocity for platooning. The network cannot have rich dynamics due to safety constraints, and perturbations from the desired platooning velocity should be very small. Understanding these networks is key for traffic management, and can form a first step in predicting traffic jams and accidents. Other applications with similar conditions include multi-satellite arrays, UAVs, drones, and robots. 

\section{Time Complexity Bounds for the Network Identification Problem} \label{sec.Complexity}
In the previous sections we presented an algorithm solving the network identification problem in $O(n^{\omega_1})$ time using specially designed inputs. We ask ourselves if we can improve on that. We first need to discretize our problem in order to fit it into the standard complexity theory framework. The presented algorithm, Algorithm \ref{alg.General}, only measured the output of the network, but we consider more general algorithms which can also measure the derivatives of the output.
\begin{prob}\label{prob.DiscreteNetworkReconst}
We are given a diffusively coupled network $(\G_\nu,\Sigma,g)$ for known agents $\Sigma$ and static controllers $g$. We are also given an integer $q>0$, such that if the input to the network is a $\mathcal{C}^{q+1}$ signal, then the output is a $\mathcal{C}^q$ signal\footnote{This is weaker than assuming that the functions $f,q,g,h$ are all smooth.}. Find the weighted graph $\G_\nu$ using measurements of the node outputs $y(t)$ and their derivatives $\frac{d^k}{dt^k}y(t)$ up to order $q$. The exogenous input signal $w(t)$ can be chosen as any $\mathcal{C}^{q+1}$ signal. Furthermore, accessing the measurements $y(t)$ or changing the function describing $w(t)$ can not be performed faster than at $\Delta t$-second intervals. Moreover, the measured outputs $y(t)$ are accurate up to a relative error no larger than $\e$.
\end{prob}

After discretizing the problem, limiting the rate of measurement and change in the input, we prove the following theorem.
\begin{thm}\label{thm.Complexity}
Any (possibly randomized) algorithm solving Problem \ref{prob.DiscreteNetworkReconst}, estimating $\{\nu_{ij}\}$ with some finite error (with probability $1$), must make $n-1$ measurements in the worst case. Moreover, if the algorithm is deterministic, its worst-case complexity is at least $\Omega(n^{\omega_1})$.
\end{thm}

\begin{cor}
By Remark \ref{rem.Omega1Implementation2}, Algorithm \ref{alg.General} is optimal in terms of computational time complexity
\end{cor}

The proof of the theorem relies on two lemmas:
\begin{lem} \label{lem.Numerics}
Let $P\in \mathbb{R}^{m\times m}$ be a positive definite matrix, let $\varrho > 0$, and let $w(t)$ be a $ \mathcal{C}^{q+1}$ signal. Consider the system $\Sigma_P : \dot{x} = -\varrho P x + w, y=x$. If $\varrho \gg \frac{1}{\underline{\sigma}(P)\Delta t}$ then for any $0\le j \le q$ and any time $T\ge \Delta t$, the equality $\frac{d^j y}{dt^j}(T) =  \frac{1}{\varrho} P^{-1} \frac{d^j w}{dt^j}(T)$ holds up to a relative error no larger than $\e$.
\end{lem}
\begin{lem} \label{lem.Structure}
Let $P\in \mathbb{R}^{(n-1)\times(n-1)}$ be a positive definite matrix, let $\E_{K_n}$ be the incidence matrix of the complete graph on $n$ edges, and let $V\in \mathbb{R}^{(n-1)\times n}$ be any matrix such that $VV^\top = \Id_{n-1}$ and $V\mathbbm{1}_n = 0$. There exists a positive semi-definite matrix $\mathcal Q\in \mathbb{R}^{n\times n}$ and a positive-definite diagonal matrix $N$ such that $\mathcal{Q} = \E_{K_n} N \E_{K_n}^\top$ and $P = V\mathcal Q V^\top$.
\end{lem}
The proof of Lemma \ref{lem.Numerics} is very technical and is relegated to Appendix \ref{append.Numerics}. We now prove Lemma \ref{lem.Structure}.
\begin{proof}
Define $\mathcal Q = V^\top PV\in\mathbb{R}^n$, which is positive semi-definite as $P$ is positive definite. Moreover, we have:
$$
V\mathcal Q V^\top = VV^\top P VV^\top = \Id_{n-1} P \Id_{n-1} = P
$$
which proves the second claim. As for the first, define the matrix $N$ as follows - for each edge $e=\{i,j\}$ in $K_n$, define the $e$-th diagonal entry of $N$ as $-\mathcal Q_{ij}=-\mathcal Q_{ji}$. Note that the off-diagonal entries of $\mathcal Q$ are equal to the off-diagonal entries of $\E_{K_n}N\E_{K_n}^\top$, as the latter is a weighted Laplacian. As for the diagonal entries, $\mathbbm{1}_n$ is in the kernel of both $\E_{K_n}N\E_{K_n}^\top$ and $\mathcal Q = V^\top PV$. Thus the sum of the elements in each row of both matrices is zero, meaning that:
$$
\mathcal Q_{ii} = -\sum_{j\ne i} \mathcal Q_{ij},\ (\E_{K_n}N\E_{K_n}^\top)_{ii} = -\sum_{j\ne i} (\E_{K_n}N\E_{K_n}^\top)_{ij}.
$$ \normalsize
Therefore the diagonal entries are also equal. This implies that $\mathcal Q = \E_{K_n}N\E_{K_n}^\top$ and completes the proof of the lemma.
\end{proof}

We now prove Theorem \ref{thm.Complexity}.
\begin{proof}
We first deal with a similar problem. We consider a single agent with $m$ inputs and $m$ outputs, evolving according to the equation $\dot{x} = -f(x) + w,\ y = h(x)$. We are again allowed to measure the output and its derivatives up to order $q$, or change the $\mathcal{C}^{q+1}$ function defining the input, no faster than once every $\Delta t$ seconds. Moreover, all measurements are accurate up to a relative error of $\e$.  Specifically, we choose any positive-definite matrix $P\in \mathbb{R}^{m\times m}$ and an arbitrary large enough scalar $\varrho>0$, and consider the single agent $\Sigma_P$ with $m$ inputs and outputs, as defined in Lemma \ref{lem.Numerics}. We claim any algorithm computing $P$ with some finite error (with probability $1$) must take at least $m$ measurements in the worst case, and that if the algorithm is deterministic, then its worst case complexity is at least $\Omega(m^{\omega_1})$. We will prove it below, but first show this claim proves the theorem. Consider a network identification problem with agents $\dot{x}_i = u_i$, static controllers $g_{ij}(x) = x$, and an underlying graph $\G = K_n$., where the coupling matrix $N$ is unknown. The dynamics of the network can be written as $\dot{x} = -\E_{K_n}N\E_{K_n}^\top x+ w$.
We note that this system has two decoupled subsystems - one for the scalar $\mathbbm{1}_n^\top x$, and one for the relative states vector $\mathrm{Proj}_{\mathbbm{1}_n^\perp}x$. Focusing on the latter, we consider the matrix $V\in\mathbb{R}^{(n-1)\times n}$ having the following vectors as rows for $k=1,\ldots,n-1$:
$$
\textcolor{black}{v^{(k)}} = \frac{1}{\sqrt{k^2+k}} [\underset{\text{$k$ times}}{\underbrace{1,\cdots,1}},-k,\underset{\small\text{$n-k-1$ times}\normalsize}{\underbrace{0,\cdots,0}}]
$$
It is easy to check that $V^\top V = \mathrm{Id}_{n-1}$ and that $\mathbbm{1_n}\in \ker(V)$. The vector $z = Vx$ satisfies the ODE $\dot{z} = -V\E_KN\E_K^\top V^\top z + Vw$. By Lemma \ref{lem.Structure}, we get a general system of the form $\Sigma_P$, where $P$ can be any positive definite matrix, and reconstructing $P = V\E_KN\E_K^\top V^\top$ is equivalent to reconstructing $N$. This completes the proof of the theorem, as here $m = n-1$.

Now, return to the system identification problem for the system $\Sigma_P$. Consider any measurement made by the algorithm. 
Suppose that at time $T_1$ we measured the $\ell$-th derivative of the output, and let $T_0$ be the last time the function describing the input $w(t)$ was changed. Noting that $T_1 - T_0 \ge \Delta t$, we conclude that if $\varrho$ is large enough then the measurement is equal to $\frac{1}{\varrho} P^{-1} \frac{d^\ell w}{dt^\ell}(T_1)$ up to a relative error of size $\e$. In particular, if $\varrho$ is sufficiently large, all measurements made by the algorithm will be of the form \textcolor{black}{$\mathrm z^{(i)} = \frac{1}{\varrho}P^{-1}\uptau^{(i)}$}, where the vector \textcolor{black}{$\uptau^{(i)}$} depends on $w(t)$ and can be calculated exactly.  We first assume that less than $m$ measurement were made. Let \textcolor{black}{$\mathrm z^{(1)},\ldots,\mathrm z^{(r)}$} be the measurements and \textcolor{black}{$\uptau^{(1)},\ldots,\uptau^{(r)}$} be the corresponding functions of the inputs. If $r<m$, we can find a nonzero vector $\uptau_\star$ which is orthogonal to \textcolor{black}{$\uptau^{(1)},\ldots,\uptau^{(r)}$}. This means that the systems $\Sigma_P$ and $\Sigma_{\mathcal P_\alpha}$ where $\mathcal P_\alpha=(P^{-1}+\alpha \uptau_\star \uptau_\star^\top)^{-1}$ will yield the same measurements, and we cannot differentiate between them. Moreover, the error can be arbitrarily large for different values of $\alpha$. Thus any (possibly randomized) algorithm solving the problem, estimating the system matrix $\varrho P$ up to some finite error with probability 1, should take at least $m$ measurements. Now, the relationship between $\frac{1}{\varrho}\textcolor{black}{\uptau^{(i)}}$ and $\textcolor{black}{\mathrm z^{(i)}}$ is linear, with the connecting matrix being $P^{-1}$, so taking more than $m$ measurements does not yield any additional data. Hence, the algorithm takes measurements of $P^{-1}$ times $m$ vectors, and returns the value of $P$. Thus, it solves the matrix inversion problem for positive-definite matrices, and thus has complexity of at least $\Omega(m^{\omega_1})$.
\end{proof}

\vspace{-10pt}
\section{Numerical Examples} \label{sec.CaseStudy}
We now apply Algorithm \ref{alg.General} in two examples. The first considers a network of oscillators subject to dry friction, and the second considers an opinion dynamics model in which the agents converge to consensus in finite-time. \textcolor{black}{In both cases, the steady-state relations are smooth, but the dynamics are discontinuous. In particular, Algorithm 1 is applicable although other works in the literature concerned with nonlinear agents, including \cite{Gardner2003,Guangjun2015,Prasse2020,Chen2009,Burbano2019}, are inapplicable.}
\begin{figure*}[!t]
\begin{center}
	\subfigure[Adjacency Matrix of the Graph in the First Case Study. Yellow entries are equal to 1, blue entries are equal to 0.] {\scalebox{.46}{\includegraphics{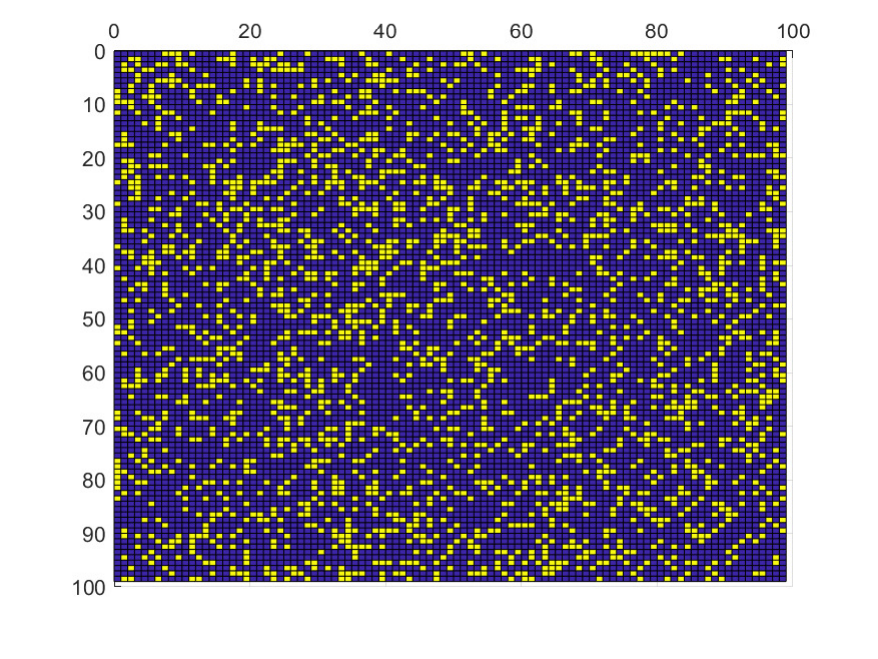}}\label{fig.LTIGraph}} \hfill
	\subfigure[Estimation Errors of Coupling Strengths, as achieved by Algorithm \ref{alg.General}]{\scalebox{.46}{\includegraphics{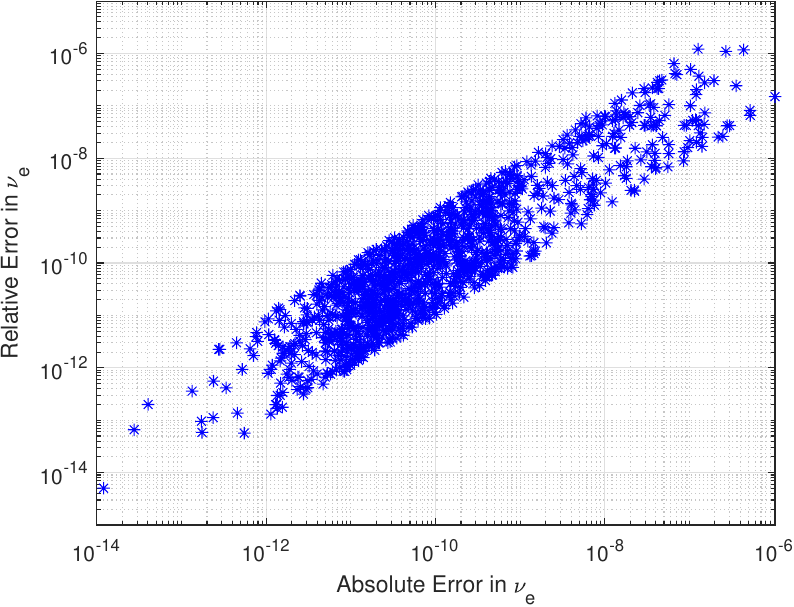}}\label{fig.LTIErrors}}
	\vspace{-5pt}
  \caption{Network identification of a network of oscillators with dry friction.}
  \vspace{-15pt}
\end{center}
\end{figure*}

\vspace{-10pt}
\subsection{Oscillators with Dry Friction}
We consider a network of $n=100$ oscillators, affected by dry friction \cite{vanderLinden1993}. These are governed by the equation $\ddot{x}_i + F_i{\rm sgn}(\dot{x}_i) + \omega_i^2 x_i = u_i$, where $x_i$ is the position of the mass, $F_i$ is the dry friction term, $\omega_i$ is the undamped angular frequency, and $u_i$ is the input. The oscillators form a diffusively-coupled MAN on a graph $\G$, in which each edge appears with probability $p=0.25$ independently from all other edges, and the edge controllers are static unit gains. We note that the agent dynamics are not smooth, as the ${\rm sgn}$ function is discontinuous at $0$. However, the steady-state relation for the agents is given by $k^{-1}_i(\mathrm y_i) = \omega_i^{-2}\mathrm y_i$, which is smooth.

For our simulation, the parameters $F_i,\omega_i$ were chosen log-uniformly between $1$ and $10$, and the unknown coupling coefficients $\nu_{ij}$ were chosen log-uniformly between $0.1$ and $10$. Algorithm \ref{alg.General} was run with $\epsilon = 0.01$ and $\kappa = 0.1$. Instead of waiting for convergence, the switching signal changed its value every 100 seconds. The adjacency matrix of the randomly chosen graph is available in Fig. \ref{fig.LTIGraph}. The algorithm correctly identified all edges existing in $\G$, and Fig. \ref{fig.LTIErrors} shows the absolute and relative errors calculating the weights $\nu_{ij}$. The maximal absolute error is about $\num{1.02e-6}$, and the maximal relative error is about $\num{1.21e-6}$.

\vspace{-10pt}
\subsection{Opinion Dynamics and Finite-Time Consensus}
We consider a collection of $n=100$ agents implementing the finite-time consensus protocol for opinion dynamics appearing in \cite{Shi2018}, given by $\dot{x}_i = c_i{\rm sgn}\left(-\sum_{j=1}^n l_{ij}x_{j}\right) + w_i$, where $c_i$ represents the conformity of the $i$-th agent, and $l_{ij}$ represent the strength of the relationship between agents $i$ and $j$. The matrix $L = (l_{ij})$ is assumed to be a graph Laplacian of some (undirected) graph $\G$, $L = \E_\G N \E_\G^\top$. This is a diffusively-coupled MAN with agents $\Sigma_i: \dot{x}_i = {\rm sgn}(u_i), y_i = x_i$ and static gain controllers equal to $\nu_{ij}$, i.e. $g_{ij}$ is the identity map. The dynamics of the network are discontinuous, as ${\rm sgn}$ is discontinuous at $0$. However, the steady-state relations of the agents are given by $k^{-1}_i (\mathrm y_i) = 0$, and are therefore smooth.

For our simulation, we let $\G$ be a random graph on $n=100$ nodes, in which each edge appears with probability $p=0.15$, independently from all other edges. The parameters $c_i$ were chosen log-uniformly $0.1$ between $10$, and the weights $\nu_{ij}$ were chosen log-uniformly between $1$ and $10$. 
Algorithm \ref{alg.General} was run with $\kappa = 0.1$ and $\epsilon = 0.01$. As with the previous case study, we did not wait for convergence, but instead changed the value of the switching signal every $50$ seconds. The adjacency matrix of the randomly chosen graph is available in Fig. \ref{fig.LTIGraph}. The algorithm correctly identified all edges existing in $\G$, and Fig. \ref{fig.LTIErrors} shows the absolute and relative errors calculating the weights $\nu_{ij}$. The maximal absolute error is about $\num{1.02e-5}$, and the maximal relative error is about $\num{2.92e-6}$. 
\begin{figure*}[!t]
\begin{center}
	\subfigure[Adjacency Matrix of the Graph. Yellow entries are equal to 1, blue entries are equal to 0.] {\scalebox{.46}{\includegraphics{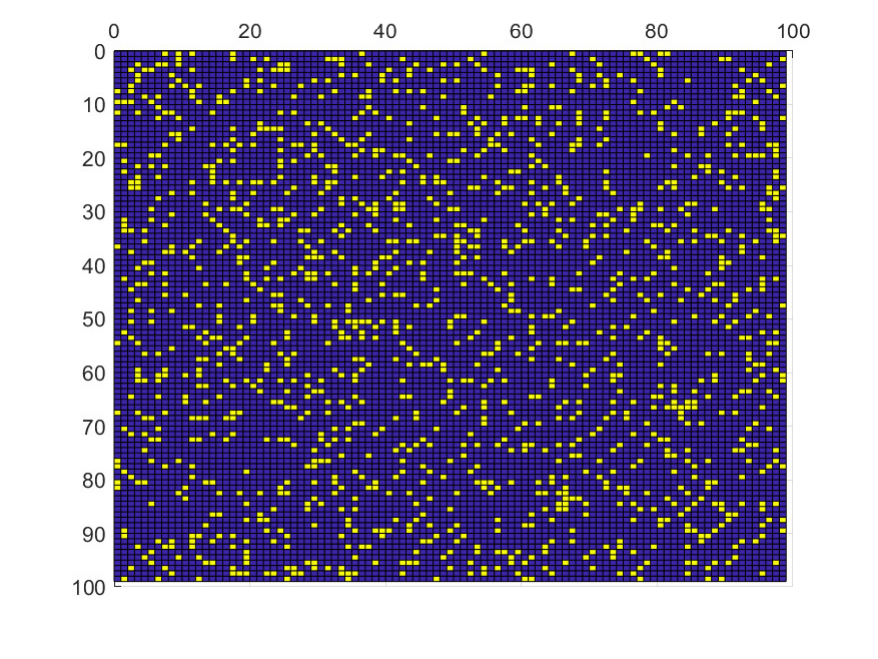}}\label{fig.NNGraph}}\hfill
	\subfigure[Estimation Errors of Coupling Strengths, as achieved by Algorithm \ref{alg.General}]{\scalebox{.46}{\includegraphics{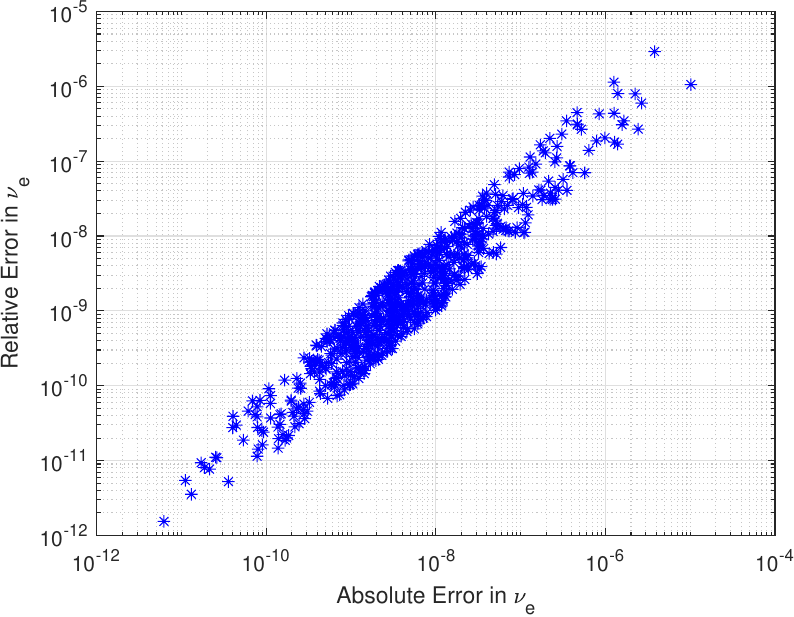}}\label{fig.NNErrors}}
  \caption{Network identification for a network of agents running a finite-time consensus protocol.}
  \vspace{-15pt}
\end{center}
\end{figure*}

\vspace{-5pt}
\section{Conclusion}
We presented a network identification algorithm using probing inputs, with no prior knowledge on the network but only on the agents and the controllers. This was done by injecting a prescribed switching signal, achieved for globally converging networks, allowing identification of the underlying network in a very general case. The resulting algorithm had sub-cubic time complexity. We discussed the different assumptions of the algorithm, and presented a lower bound on the complexity of any algorithm solving the network identification problem, proving that the presented algorithm is optimal in sense of time complexity. We demonstrated the results in simulation, showing the algorithm can be applied for large networks with non-smooth dynamics. \textcolor{black}{While the analysis was carried out for the case of SISO agents for the sake of simplicitly, the ideas all generalize easily to MIMO agents with equal input and output dimension, with minor changes to Algorithm \ref{alg.General}.}

\vspace{-5pt}
\bibliographystyle{ieeetr}
\bibliography{root}

\appendices
\section{Proof of Theorem \ref{thm.GeneralIndependence}} \label{append.ProofIndependence}

We start by stating and proving a lemma, which will allow us to assume that $k^{-1}$ and $g$ are both differentiable at \textcolor{black}{$\mathrm w^{(0)}$}.
\begin{lem}
Suppose that the same assumptions as in Theorem \ref{thm.GeneralIndependence} hold. Then for any $i\in\{1,\ldots,n\}$ and any number $\mathrm x \in \mathbb{R}$, the \textcolor{black}{ set $\mathcal{S}$} of all $\mathrm w\in \mathbb{R}^n$ such that the solution $\mathrm y$ to $\mathrm w = k^{-1}(\mathrm y) + \E_\G N g(\E_\G^\top\mathrm y)$ satisfies $\mathrm{y_i = x}$ has measure zero.
\end{lem}

\begin{proof}
We consider the map $G:\mathbb{R}^n\to \mathbb{R}^n$ defined by $G(\mathrm y) = k^{-1}(\mathrm y) + \E_\G N g(\E_\G^\top\mathrm y)$. The relevant set $\mathcal{S}$ is the image of \textcolor{black}{
$\mathcal{R} = \{\mathrm y \in \mathbb{R}^n : \mathrm{y_i = x}\}$} under $G$. The assumption on $k^{-1},g$ implies that $G$ is continuous and piecewise smooth, hence locally Lipschitz. Thus, $G$ is absolutely continuous, sending zero-measure sets to zero-measure sets. As $\mathcal{R}$ has measure zero, we conclude that $\mathcal{S}$ also has measure zero.
\end{proof}

\begin{cor} \label{cor.GeneralPoint}
Under the same assumptions as in Theorem \ref{thm.GeneralIndependence}, with probability $1$, the functions $k^{-1}$ and $g$ are twice differentiable at \textcolor{black}{$\mathrm y^{(0)}, \E_\G^\top \mathrm y^{(0)}$} respectively, and the differential $\nabla g$ is a positive-definite diagonal matrix.
\end{cor}

We can now prove Theorem \ref{thm.GeneralIndependence}.
\begin{proof}
By Corollary \ref{cor.GeneralPoint}, we can assume $k^{-1}$ is twice differentiable at \textcolor{black}{$\mathrm y^{(0)}$}. Under this assumption, we can write the following equation connecting $\delta \textcolor{black}{\mathrm y^{(i)}}$ and $\delta \textcolor{black}{\mathrm w^{(i)}}$ for $i=1,2,\ldots,n$,

\textcolor{black}{
\vspace{-8pt}
\small
\begin{align*} %\label{eq.AppendixConnection}
\delta \mathrm w^{(i)} = \left[\nabla k^{-1}(\mathrm y^{(i)}) + \E_\G N \nabla g(\E_\G^\top \mathrm y^{(0)}) \E_\G^\top\right] \delta \mathrm y^{(i)} + O(\|\delta \mathrm y^{(i)}\|^2),
\end{align*}\normalsize}
which follows from Proposition \ref{prop.QuadraticOutputError} in the case $\nabla k^{-1}(\textcolor{black}{\mathrm y^{(0)}}) \neq 0$ or $i<n$, and uses $\E_\G^\top \mathbbm{1}_n = 0$ otherwise. Because $\kappa$ is small, and $k^{-1}$ and $g$ are twice differentiable at \textcolor{black}{$\mathrm y^{(0)}$ and $\E_\G^\top \mathrm y^{(0)}$}, we can conclude that \textcolor{black}{$\|\delta \mathrm y^{(i)}\| = O(\|\delta \mathrm w^{(i)}\|)$}. Thus, recalling the definition of $\mathcal{M}$, we can rewrite the previous equation as:

\textcolor{black}{
\vspace{-10pt}
\small
\begin{align} \label{eq.12}
\delta \mathrm w^{(i)} - O(\|\delta \mathrm w^{(i)}\|^2) = \mathcal{M}\delta \mathrm y^{(i)}.
\end{align}
\normalsize
\vspace{-10pt}}

We first focus on the case $\nabla k^{-1} (\textcolor{black}{\mathrm y^{(0)}}) \neq 0$. The matrix $\mathcal{M}$ is invertible, as Assumption \ref{Assumption2} implies that $\nabla k^{-1} \ge 0,\nabla g \ge 0$ \cite{Sharf2018b}. Thus $\mathcal{A}$ is linearly independent if and only if the vectors on the left-hand side of \eqref{eq.12}, $\kappa \mathrm e_i - \textcolor{black}{\mathrm z^{(i)}}$ for some vectors $\textcolor{black}{\mathrm z^{(i)}}$ satisfying $||\textcolor{black}{\mathrm z^{(i)}}|| = O(\kappa^2)$, are linearly independent. Thus, they are linearly independent if $\kappa$ small, and $\mathcal{A}$ is a basis.

As for the case $\nabla k^{-1}(\textcolor{black}{\mathrm y^{(0)}}) = 0$, we note $ \E_\G N \nabla g(\E_\G^\top \textcolor{black}{\mathrm y^{(0)}}) \E_\G^\top$ preserves the space orthogonal to $\mathbbm{1}_n$. As \textcolor{black}{$\delta \mathrm y^{(1)},\ldots,\delta \mathrm y^{(n-1)}$} are orthogonal to $\delta \textcolor{black}{\mathrm y^{(n)}} = \kappa \mathbbm{1}_n$, it is enough to show that the former are linearly independent. Moreover, as the map $ \E_\G N \nabla g(\E_\G^\top \textcolor{black}{\mathrm y^{(0)}}) \E_\G^\top$ is invertible on the space $\mathbbm{1}_n^\perp$, it is enough to prove that the vectors on the left hand side of equation \eqref{eq.12} are linearly independent. However, these vectors are of the form $\kappa(\mathrm e_i-\mathrm e_n) - O(\kappa^2)$, which are linearly independent if $\kappa$ is small enough, similarly to the first case. Thus $\mathcal{A}$ is a basis. This concludes the proof.
\end{proof}

\begin{rem} \label{rem.QuadraticInputError}
Note that we used the twice differentiability assumption to get \textcolor{black}{$\delta\mathrm y^{(i)} = O(\|\delta \mathrm w^{(i)}\|)$}. In particular, the error rate in Proposition \ref{prop.QuadraticOutputError} is $O(\|\delta \textcolor{black}{\mathrm w^{(i)}}\|^2)$. 
\end{rem}

\section{Proofs of Error Bounds} \label{append.ProofAlgorithm}
This appendix is dedicated to proving Theorem \ref{thm.GeneralErrorRateAndComplexity} and Proposition \ref{prop.MeasError}. We start by stating and proving a lemma, which will allow us to efficiently bound the term $\|\delta Y^{-1}\|$.
\begin{lem} \label{lem.InvW}
Let $\delta W$ be the matrix computed by Algorithm \ref{alg.General}. Then it is the product of no more than $O(n)$ elementary matrices, and the operator norm of $\delta W^{-1}$ is bounded by $2/\kappa$.
\end{lem}

\begin{proof}
If $\nabla k^{-1} (\textcolor{black}{\mathrm y^{(0)}}) \neq 0$, then $\delta W=\kappa \mathrm{Id}_n$, and the result is clear.
If $\nabla k^{-1} (\textcolor{black}{\mathrm y^{(0)}}) = 0$, then $\delta W$ is equal to $\kappa F$, where the columns of $F$ are given by $\mathrm e_i - \mathrm e_n$ for $i=1,\ldots,n-1$ and $\mathbbm{1}_n$. Thus $(\delta W)^{-1} = \kappa^{-1}F^{-1}$, and it suffices to show that the operator norm of $F^{-1}$ is bounded by $2$, and that $F$ is the product of no more than $O(n)$ elementary matrices.

We run a Gaussian elimination procedure on the matrix $F$. Each row operation corresponds to multiplication by an elementary matrix, so it suffices to show that the procedure halts after $O(n)$ steps. We first consider row operations of the form $R_n \to R_n + R_i$ for $i=1,\ldots,n-1$, i.e. adding row $i$ to row $n$. These are $n-1$ total row operations, leaving all first $n-1$ rows unaltered, and changing the last row of the matrix to $\left[\begin{smallmatrix}0&\cdots&0&n\end{smallmatrix}\right]$. We now divide the $n$-th row by $n$, which is another row operation, altering the last row to $\left[\begin{smallmatrix}0&\cdots&0&1\end{smallmatrix}\right]$. Lastly, we apply the row operations $R_i \to R_i - R_n$ for $i=1,\ldots,n-1$. These operations nullify all nonzero off-diagonal elements, achieving an identity matrix. Thus,we transformed the matrix $F$ to the identity matrix by applying a total of $2n-1$ row operations, hence $F$ is the product of $2n-1$ elementary matrices.

We now compute $F^{-1}$ by applying the previous row operations, which transformed $F$ to ${\rm Id}_n$, in order to transform ${\rm Id}_n$ to $F^{-1}$. First, we applied the row operations $R_n \to R_n + R_j$ for $j=1,\ldots,n-1$. This leaves all rows but the last unaltered, and the last row becomes $\left[\begin{smallmatrix}1&\cdots&1\end{smallmatrix}\right]$. Then, we applied the row operation $R_n \to \frac{1}{n} R_n$, dividing the last row by $n$. Lastly, we applied the row operations $R_j \to R_j - R_n$ for $j=1,\ldots,n-1$, subtracting $\frac{1}{n}\left[\begin{smallmatrix}1&\cdots&1\end{smallmatrix}\right]$ from each row but the last. Thus, the matrix $F^{-1}$ is the sum of two matrices, $\frac{1}{n}\upxi\mathbbm{1}_n^\top$, where $\upxi = [-1,\ldots,-1,1]^\top$, and the diagonal matrix $\mathcal{I} = \mathrm{Id}_n - \mathrm e_n \mathrm e_n^\top$, having all diagonal entries equal to $1$, but the last, which is equal to $0$. We conclude
\begin{align*}
\|F^{-1}\| \le \frac{1}{n} \|\upxi\mathbbm{1}_n^\top\| + \|\mathcal{I}\| \le \frac{1}{n}\|\mathbbm{1}_n\|\|\upxi\| + 1 = 1+1 = 2,
\end{align*}
completing the proof.
\end{proof}

We can now prove Theorem \ref{thm.GeneralErrorRateAndComplexity}:
\begin{proof}
We start by proving the floating point operations estimate. The first part of the algorithm, before the for-loop, takes $O(n^2)$ time, mostly to initialize $\delta \textcolor{black}{\mathrm w^{(i)}}$ for $i=1,\ldots,n$. The first for-loop takes $O(n^2)$ time, as in each iteration we store $\textcolor{black}{\mathrm y^{(i)}}$ in memory and multiply \textcolor{black}{$J(\mathrm y^{(i)} - \mathrm y^{(0)})$}, each taking $O(n)$ time, as $J$ is either an identity matrix or a projection on  $\mathbbm{1}_n^\perp$. The computation between the two for-loops takes $O(n^\omega)$ time. Lastly, the last for-loop also takes $O(n^2)$ time. The result is now obtained since $\omega \ge 2$.

We now prove the error bound. By definition, the matrix $M^{\prime}$ satisfies \textcolor{black}{$M^{\prime}\delta \textcolor{black}{\mathrm y^{(i)}} = \delta \mathrm w^{(i)}$} for $i =1\ldots,n$. Using Corollary \ref{cor.GeneralPoint} and Remark \ref{rem.QuadraticInputError}, we conclude that for any $i\in \V$, \textcolor{black}{$\|\delta \mathrm w^{(i)} - \mathcal M^{\prime} \mathrm \delta \mathrm y^{(i)}\| = O(\|\delta\mathrm w^{(i)}\|^2)$}, where $\mathcal{M}^\prime = \mathcal M + Q$. Thus, for any $i\in \V$ we have \textcolor{black}{$\|(M^{\prime}-\mathcal M^{\prime})\delta \mathrm y^{(i)}\|\le O(\max_i \|\delta \mathrm w^{(i)}\|^2) = O(\kappa^2)$}, and we conclude that:
\begin{align*}
\|(M^{\prime}-\mathcal M^{\prime})\delta Y\| \le \sqrt{\sum_{i\in \V}O(\kappa^2)^2} = O(\sqrt{n}\kappa^2).
\end{align*}

By submultiplicativity of the operator norm, we conclude that  $\|M^{\prime}-\mathcal M^{\prime}\| \le \|(M^{\prime}-\mathcal M^{\prime})\delta Y\|\|\delta Y ^{-1}\|$, implying that $\|M^{\prime}-\mathcal M^{\prime}\| \le  O\left(\sqrt{n}\kappa^2\|\delta Y^{-1}\|\right)$. As $M^{\prime} - \mathcal M^{\prime} = M - \mathcal M$, we conclude that $\|M-\mathcal{M}\| \le O\left(\sqrt{n}\kappa^2\|\delta Y^{-1}\|\right)$, yielding the same bound for all entries $|M_{ij} - \mathcal M_{ij}|$.

We now wish to estimate $\|\delta Y^{-1}\|$. We define \textcolor{black}{$\delta \mathrm v^{(i)} = \mathcal M^\prime \delta \mathrm y^{(i)}$} for $i=1,\ldots,n$, so equation \eqref{eq.12} reads \textcolor{black}{$\delta \mathrm v^{(i)} = \delta \mathrm w^{(i)} - O(\kappa^2)$}. Define \textcolor{black}{$\delta V = \left[\begin{smallmatrix} \delta \mathrm v^{(1)}&\cdots&\delta \mathrm v^{(n)}\end{smallmatrix}\right]$}, so that $\delta V = \delta W - O(\kappa^2)$. By multiplying both sides by $\delta W^{-1}$ and using Lemma \ref{lem.InvW}, we conclude that $\delta W^{-1}\delta V = \mathrm{Id}_n - O(\kappa)$, or equivalently, $\delta V^{-1}\delta W = \mathrm{Id}_n + O(\kappa)$. As $\delta Y^{-1} = \mathcal{M}^\prime \delta V^{-1}$, we can again use the submultiplicativity of the operator norm:
\begin{align*}
\|\delta Y^{-1}\| \le \|\mathcal M^\prime\| \cdot \|\delta V^{-1}\| \le \|\mathcal M^\prime\| \cdot \|\delta V^{-1}\delta W\| \cdot \|\delta W^{-1}\|.
\end{align*} 
We can now estimate each factor on its own. Lemma \ref{lem.InvW} implies that $\|\delta W^{-1}\| = O(\kappa^{-1})$. Moreover, $\delta W^{-1}\delta V = \mathrm{Id}_n - O(\kappa)$ implies that $\|\delta V^{-1} \delta W\| = 1+O(\kappa)$. Lastly, we have:
\begin{align*}
\|\mathcal M^\prime\| \le&~ \|\nabla k^{-1}(\textcolor{black}{\mathrm y^{(0)}})\| + \|\E_\G N \nabla g(\E^\top_\G \textcolor{black}{\mathrm{y}^{(0)}}) \E_\G^\top\| + \|Q\| \\ \le&~ O(1) + \max_{i,j} (\nu_{ij}d_{ij})\lambda_{\text{max}}(\G) + O(1),
\end{align*}
where $\|Q\|\le 1$ in both cases $Q = 0$ and $Q = \frac{1}{n}\mathbbm{1}_n\mathbbm{1}_n^\top$. 
\end{proof}

We now move to proving Proposition \ref{prop.MeasError}:
\begin{proof}
It's enough to bound $\|\delta W (\delta Y^{-1} - (\delta Y + \Delta Y)^{-1})\|$. Using the submultiplicativity of the operator norm, we can bound each factor on its own. For the first factor, we note that $\delta W$ is either equal to $\kappa \mathrm{Id}_n$, or to $\kappa(\mathrm{Id}_n - \mathrm e_n \mathbbm{1}_n^\top + \mathbbm{1}_n \mathrm e_n^\top)$, depending on whether $\nabla k^{-1}(\textcolor{black}{\mathrm y^{(0)}}) = 0$ or not. In both cases, $\|\delta W\| \le \kappa \sqrt{n}$. As for the second factor, we can bound it as $\|\delta Y^{-1}\|\|\mathrm{Id_n} - (\mathrm{Id}_n + \Delta Y \delta Y^{-1})^{-1}\|$. The assumption $||\Delta Y|| \ll ||\delta Y||$ implies that $(\mathrm{Id}_n + \Delta Y \delta Y^{-1})^{-1} \approx \mathrm{Id}_n - \Delta Y \delta Y^{-1}$ up to a second order error, which in turn gives the desired bound, as $\|\delta Y^{-1}\| \le O(\kappa^{-1}(1+\max_{i,j}(\nu_{ij}d_{ij})\lambda_{\mathrm{max}}(\G)))$, see the proof of Theorem \ref{thm.GeneralErrorRateAndComplexity}.
\end{proof}

\section{Proof of Lemma \ref{lem.Numerics}} \label{append.Numerics}
\begin{proof}
The output $y(t)$ at any time $t$ can be written as a convolution integral, $y(t) = \int_{0}^t e^{-\varrho \xi P} w(t-\xi)d\xi$.
By assumption, $w$ is continuously differentiable $q+1$ times, so by Lagrange's form of the remainder in Taylor's theorem, we have
\begin{align} \label{eq.wTaylorSeries}
&w(t-\xi) = \sum_{j=0}^q (-1)^j\frac{d^j w(t)}{dt^j} \frac{\xi^j}{j!} + \frac{d^{q+1}w(\tilde{t})}{dt^{q+1}}  \frac{(-\xi)^{q+1}}{(q+1)!},
\end{align}
for some point $\tilde{t}\in[t-\xi,t]$. We plug this expression inside the integral describing $y(t)$. For the first $q$ summands, we end up with integrals for the form $\int_{0}^t e^{-\varrho \xi P} \frac{d^jw(t)}{dt^j} \frac{\xi^j}{j!}$. The following formula will be used to compute them \cite[Formula 2.321.2]{Zwillinger2014}:
\begin{align} \label{eq.IntegralOfGammaFunction}
\int x^n e^{cx} dx = e^{cx}\sum_{i=0}^n (-1)^{n-i} \frac{n!}{i! c^{n-i+1}} x^i + {\rm constant},
\end{align}

The matrix $P$ is positive-definite, so we can write it as $P = \sum_{k=1}^m \lambda_k \mathrm v_k \mathrm v_k^\top$, where $\lambda_k>0$ are $P$'s eigenvalues and $\mathrm v_k$ are its eigenvectors satisfying $||\mathrm v_k|| = 1$. For any $\xi,\varrho$, we have that $e^{-\varrho\xi P} = \sum_{k=1}^m e^{-\varrho\lambda_k\xi} \mathrm v_k \mathrm v_k^\top$. Thus, we have that:
\begin{align*}
&\int_{0}^t e^{-\varrho \xi P} \frac{d^jw(t)}{dt^j} \frac{\xi^j}{j!}d\xi
=\sum_{k=1}^m\frac{\mathrm v_k \mathrm v_k^\top}{j!} \frac{d^jw(t)}{dt^j} \int_{0}^t e^{-\varrho\lambda_k \xi} \xi^j d\xi = 
\\&  \sum_{k=1}^m\frac{\mathrm v_k \mathrm v_k^\top}{j!} \frac{d^jw(t)}{dt^j} \left[  e^{-\varrho\lambda_k\xi}\sum_{i=0}^j  \frac{(-1)^{j-i}j!}{i! (-\varrho\lambda_k)^{j-i+1}} \xi^i\right]_{\xi = 0}^{t} =
\\& \sum_{k=1}^m\frac{\mathrm v_k \mathrm v_k^\top}{j!} \frac{d^jw(t)}{dt^j} \left[-e^{-\varrho\lambda_k\xi}\sum_{i=0}^j  \frac{j!}{i! (\varrho\lambda_k)^{j-i+1}} \xi^i\right]_{\xi=0}^t =
\\& \sum_{k=1}^m{\mathrm v_k }\mathrm v_k^\top \frac{d^jw(t)}{dt^j} \left[\frac{1}{(\varrho\lambda_k)^{j+1}}-e^{-\varrho\lambda_kt}\sum_{i=0}^j  \frac{t^i}{i! (\varrho\lambda_k)^{j-i+1}}\right].
\end{align*}
Using functional calculus, we conclude the integral is equal to:
$\left[\frac{P^{-j-1}}{\varrho^{j+1}} -\sum_{i=0}^j  \frac{t^i e^{-\varrho t P} P ^{i-j-1}}{i! \varrho^{j-i+1}}\right]\frac{d^jw(t)}{dt^j}$.
By \eqref{eq.wTaylorSeries}, we get:\small
\begin{align} \label{eq.ApproxY}
y(t)\approx \sum_{j=0}^k (-1)^j \left[\frac{P^{-j-1}}{\varrho^{j+1}} -\sum_{i=0}^j  \frac{t^i e^{-\varrho t P} P ^{i-j-1}}{i! \varrho^{j-i+1}}\right]\frac{d^jw(t)}{dt^j},
\end{align}\normalsize
with an error of the form $\int_0^t e^{-\varrho\xi P}\frac{d^{q+1}w(\tilde{t})}{dt^{q+1}}  \frac{\xi^{q+1}}{(q+1)!}$. We claim that if $\varrho$ is large enough, then $y(T) = \frac{1}{\varrho}P^{-1}w(T)$ up to a relative error of magnitude no larger than $\e$.

If $T \ge \Delta t$ and $\varrho\underline{\sigma}(P) \gg \Delta t$, then
%\begin{align*}
$\frac{1}{\varrho}P^{-1}\approx \sum_{j=0}^k (-1)^j\left[\frac{P^{-j-1}}{\varrho^{j+1}} -\sum_{i=0}^j  \frac{T^i e^{-\varrho t P} P ^{i-j-1}}{i! \varrho^{j-i+1}}\right]$,
%\end{align*}
up to a relative error of magnitude no larger than $\e/2$. Indeed, the $j$-th element in the sum behaves as $O(\frac{1}{\varrho^{j+1}\underline{\sigma}(P)^{j+1}})$ plus a term decreasing exponentially fast with $\varrho$ (for fixed $P,T$). The error term in \eqref{eq.ApproxY} can also be bounded similarly - if we denote $ \mu = \max_{t\in [0,T]} \left\|\frac{d^{q+1}w(t)}{dt^{q+1}}\right\|$, then the norm of the error term in \eqref{eq.ApproxY} is bounded by:
\begin{align*}
&~ M \int_0^{T} \frac{\xi^{q+1}\|e^{-\varrho\xi P}\|}{(q+1)!}d\xi \le M\int_0^{T} \frac{\xi^{q+1}e^{-\varrho\underline{\sigma}(P)\xi}}{(q+1)!}d\xi
\\ =~ &~ \frac{M}{(q+1)!}\left[e^{-\varrho\underline{\sigma}(P)\xi}\sum_{i=0}^{q+1} \frac{(-1)^{q+1-i}(q+1)!}{i! (-\varrho\underline{\sigma}(P))^{q+1-i+1}} \xi^i\right]_{\xi = 0}^{T}
\\ =~ &~ \left[-e^{-\varrho\underline{\sigma}(P)\xi}\sum_{i=0}^{q+1} \frac{M}{i! (\varrho\underline{\sigma}(P))^{q+1-i+1}} \xi^i\right]_{\xi = 0}^{T}
\\ =~ &~ \frac{M}{(\varrho\underline{\sigma}(P))^{q+2}} -e^{-\varrho\underline{\sigma}(P)T}\sum_{i=0}^{q+1} \frac{M}{i! (\varrho\underline{\sigma}(P))^{q+1-i+1}} T^i.
\end{align*}
The first element is of order $O(\frac{1}{\varrho^{q+2}})$, and the second decays exponentially with $\varrho$ (for fixed $M,P,T$). Thus, if $\rho$ is large enough, then $y(T) = \frac{1}{\varrho} P^{-1}w(T)$, up to a relative error of order of magnitude no larger than $\e$. More specifically, this happens for any $\rho>\rho_0$, where $\rho_0$ is a threshold depending on the matrix $P$, the sample time $T\ge \Delta t$, and the signal $w(t)$ (through $M$). 

As for derivatives of $y(t)$ at $t=T$, one can use the higher-order terms of \eqref{eq.ApproxY} together with the error estimate and $\dot{y} = -\varrho Py + w$ to conclude that $\frac{dy}{dt}(T) = -\varrho P y(T) + w(T) = \frac{1}{\varrho}P^{-1} \frac{dw}{dt}(T)$ up to a relative error of magnitude no larger than $\e$, provided that $\varrho$ exceeds some threshold. Similarly, $\frac{d^j y}{dt^j} = -\varrho P \frac{d^{j-1}y}{dt^{j-1}} + \frac{d^{j-1}w}{dt^{j-1}}$ for all integers $0 \le j\le q$, allowing one to argue by induction that $\frac{d^j y}{dt^j}(T) = \frac{1}{\varrho}P^{-1} \frac{d^j w}{dt^j}(T)$ up to a relative error of magnitude no larger than $\e$, provided that $\varrho$ is large enough. This completes the proof.
\end{proof}

\end{document}